\documentclass[ba,preprint]{imsart}

\RequirePackage{amsthm,amsmath,amsfonts,amssymb, algorithm, algpseudocode, bbm}
\RequirePackage[authoryear]{natbib}
\RequirePackage[colorlinks,citecolor=blue,urlcolor=blue,backref=page,backref=page]{hyperref}
\RequirePackage{graphicx}

\pubyear{2025}
\volume{TBA}
\issue{TBA}
\firstpage{1}
\lastpage{1}

\startlocaldefs
\theoremstyle{plain}

\newtheorem{proposition}{Proposition}[section]

\theoremstyle{definition}

\theoremstyle{remark}

\endlocaldefs

\begin{document}

\begin{frontmatter}
\title{Diffusion piecewise exponential models for survival extrapolation using Piecewise Deterministic Monte Carlo}
\runtitle{Diffusion piecewise exponential models}

\begin{aug}
\author[A]{\fnms{Luke}~\snm{Hardcastle}\ead[label=e1]{luke.hardcastle.20@ucl.ac.uk}},
\author[A]{\fnms{Samuel}~\snm{Livingstone}\ead[label=e2]{samuel.livingstone@ucl.ac.uk}}
\and
\author[A]{\fnms{Gianluca}~\snm{Baio}\ead[label=e3]{g.baio@ucl.ac.uk}}
\address[A]{Department of Statistical Science,
University College London\printead[presep={,\ }]{e1,e2,e3}}
\runauthor{L. Hardcastle, S. Livingstone, and G. Baio}
\end{aug}

\begin{abstract}
The piecewise exponential model is a flexible non-parametric approach for time-to-event data, but extrapolation beyond final observation times typically relies on random walk priors and deterministic knot locations, resulting in unrealistic long-term hazards. We introduce the diffusion piecewise exponential model, a prior framework consisting of a discretised diffusion for the hazard, that can encode a wide variety of information about the long-term behaviour of the hazard, time changed by a Poisson process prior for knot locations. This allows the behaviour of the hazard in the observation period to be combined with prior information to inform extrapolations. Efficient posterior sampling is achieved using Piecewise Deterministic Markov Processes, whereby we extend existing approaches using sticky dynamics from sampling spike-and-slab distributions to more general transdimensional posteriors. We focus on applications in Health Technology Assessment, where the need to compute mean survival requires hazard functions to be extrapolated beyond the observation period, showcasing performance on datasets for Colon cancer and Leukaemia patients.
\end{abstract}

\begin{keyword}[class=MSC]
\kwd[Primary ]{62F15}
\kwd{62N01}
\kwd[; secondary ]{65C05}
\end{keyword}

\begin{keyword}
\kwd{Survival Analysis}
\kwd{Stochastic Differential Equations}
\kwd{Piecewise Deterministic Markov Processes}
\kwd{Markov Chain Monte Carlo}
\end{keyword}

\end{frontmatter}

\section{Introduction}

Survival analysis is the study of the time $Y \in \mathbb{R}_{>0}$ until an event of interest occurs, where observations $\{y_i\}_{i=1}^n$ are typically subject to some form of censoring \citep{Ibrahim2001}. This censoring is said to be administrative if survival times are automatically censored at time $y_+$, conditional on $Y > y_+$. This phenomenon is often observed in data from clinical studies, with $y_+$ corresponding to the end point of a trial or observational study. In this paper we are concerned with modelling survival data under this censoring mechanism, when the inferential objective is to characterise the hazard function, $h(y)$, over an extended time horizon $(0,y_\infty)$ for $y_\infty >> y_+$. 

To illustrate the challenges caused by administrative censoring, and motivate the problem outlined above, we will focus throughout on the estimation of mean survival,
\begin{equation}
\label{eq;Mean_survival}
    \mathbb{E}[Y] = \int_0^{y_+}\exp\left(-\int_0^yh(u)du\right)dy +  \int_{y_+}^{y_\infty}\exp\left(-\int_0^yh(u)du\right)dy.
\end{equation}
Estimation of \eqref{eq;Mean_survival} often arises in the field of Health Technology Assessment (HTA) \citep{Baio2020,Latimer2011}. In brief, HTA provides a decision-theoretic framework for analysing the cost-effectiveness of novel medical interventions in publicly funded healthcare systems. In England, following the recommendations of the National Institute for Health and Care Excellence (NICE), expected life years (i.e mean survival) is commonly incorporated as the primary measure of benefit in these analyses. Since 2018, $56\%$ of NICE appraisals for cancer treatments have been conducted using immature survival data \citep{Gibbons2024}, where the majority of events occur in $(y_+,\infty)$. Estimating \eqref{eq;Mean_survival} is therefore an increasingly important applied problem, often undertaken within a Bayesian paradigm.

Equation \eqref{eq;Mean_survival} reveals the challenges associated with administrative censoring, as accurate estimation of mean survival is reliant on correctly modelling $h(y)$ during both the observation period, $y \in (0,y_+)$, \textit{and} during the extrapolation period, $y \in [y_+,y_\infty)$. Given a sufficiently flexible or well-specified model, the data will ensure accurate inference during the former. Inference during the latter, however, will inherently be driven by model and prior specification, becoming increasingly sensitive to these choices as the rate of administrative censoring increases (up to $\approx 90\%$ as in the data in Section \ref{sec;TimeVarying}). These concerns have lead to a rise in the use of Bayesian methods for modelling \eqref{eq;Mean_survival} due to the possibility of incorporating explicit prior information to guide extrapolation \citep{Baio2020, Rutherford2020}.

Perhaps the most common approach is to assume $Y$ is generated by a (typically two-parameter) parametric survival distribution, where the parametric form allows hazards to be extrapolated beyond $y_+$ \citep{Latimer2011}. Prior information can be incorporated to inform the values of certain parameters \citep{Palmer2023}; however, this approach will typically underestimate uncertainty in the extrapolation period due to the untestable assumption that the functional form of the hazard has been correctly specified. 

More recent works have focused on flexible Bayesian survival models to overcome this issue \citep{Rutherford2020}, including spline-based approaches \citep{Jackson2023}, polyhazard models \citep{Demiris2015, Hardcastle2024} and piecewise models \citep{Cooney2023a, Kearns2019}. In each of these cases, extrapolations are either driven by behaviour of the hazard inferred during the observation period, external data included in the model, or both. Notably, the goal of using explicit prior information to guide extrapolations is rarely realised. A recent promising development in this direction is the incorporation of prior information into spline models using the Sheffield elicitation framework \citep{Jackson2023, Oakley2025}. This information, however, must be encoded as a synthetic dataset, which limits the type of prior information that can inform the model. Further, extrapolations are highly sensitive to the underlying modelling assumptions, e.g.~the placement of knots beyond $y_+$.

In this work we seek to develop a model and corresponding prior structure that allows for principled inference of $h(y)$ in both the observation and extrapolation period. We require a sufficiently flexible model during the observation period, while inferences during the extrapolation period should primarily be driven by explicit prior information. 

There are two primary considerations for specifying this prior information: \textit{i)} Assumptions about the form of this prior information should be minimal allowing the analyst maximal flexibility in its specification \citep{Mikkola2024}. \textit{ii)} The prior should be at least moderately informative during the extrapolation period. We argue, given the often sparse nature of data in these applications, that specification of an informative prior is the \textit{only} way to ensure sensible inference in the extrapolation period.

\subsection{Our contributions}
\label{sec;Contributions}
We introduce the Diffusion Piecewise Exponential Model. The piecewise exponential model is defined by a piecewise constant log-hazard function,
\begin{equation}
    \label{eq;PEM_haz}
    \log h(y) = \sum_{j=1}^J\alpha_j\mathbbm{1}\left(y \in (s_{j-1},s_j]\right),
\end{equation}
where $\{\alpha_j\}_{j=1}^J$ are a sequence of local log-hazards, and $\{s_j\}_{j=0}^J$ are a sequence of knot locations with $s_0 = 0$. Explicitly, our contributions are as follows.

In Section \ref{sec;Model}, we introduce a novel prior formulation for the sequences $\{\alpha_j\}_{j=1}^J$ and $\{s_j\}_{j=0}^J$, allowing for the principled combination of inferences for the observation period, primarily driven by the data, and inferences for the extrapolation period, primarily driven by prior information. This prior for $\{\alpha_j\}_{j=1}^J$ is given by the discretisation of a diffusion, with drift function used to encode strong prior information about the long-term behaviour of the hazard function. Notably, restrictions on the form of the drift are minimal allowing for a range of prior information to be encoded into the model. The prior for $\{s_j\}_{j=0}^J$ is given by a Poisson point process. This acts as a time change between the underlying diffusion and $\{\alpha_j\}_{j=1}^J$ allowing for intensity in the changes of the hazard during the extrapolation period to be informed by those observed on $(0,y_+)$.

In Section \ref{sec;Sampling}, we introduce a novel Markov Chain Monte Carlo (MCMC) sampling algorithm based on Piecewise Deterministic Markov Processes (PDMPs). In particular we make use of recent developments in defining and generating these processes to design an efficient sampler that requires minimal user tuning \citep{Bertazzi2023, Michel2020}. Further, to handle the transdimensional posterior resulting from the prior on $\{s_j\}_{j=0}^J$, we extend recent results that use PDMPs to sample from posteriors induced by spike and slab priors \citep{Bierkens2023a, Chevallier2023} to more general transdimensional posteriors. 

In Section \ref{sec;Examples} we demonstrate the flexibility of the model and prior structure, and provide practical guidelines for its use via case studies corresponding to two clinical data sets. We conclude with a discussion in Section \ref{sec;Discussion}.

\section{The Diffusion Piecewise Exponential Model}
\label{sec;Model}
Throughout we assume that we observe data, $\mathcal{D} = \{y_i, \delta_i, w_i\}_{i=1}^n$, consisting of $n$ independent survival times $y_i$ and covariate vectors $w_i \in \mathbb{R}^p$. These observations either correspond to the time of the event, if $\delta_i = 1$, or right censored observations such that $y_i < Y_i$, with $\delta_i = 0$. Additionally we assume that the censoring mechanism is non-informative, and that there exists an administrative censoring time $y_+$ after which all observations are censored. This set up is ubiquitous in HTA practice, as well as in many other applied fields.

Survival data are typically analysed through the hazard function, $h(y)$, and the survival function, $S(y)$, connected through the likelihood
\begin{equation*}
    \mathcal{L}(\cdot\mid\mathcal{D}) = \prod_{i=1}^nh(y_i)^{\delta_i}S(y_i) = \prod_{i=1}^nh(y_i)^{\delta_i}\exp\left(-\int_0^{y_i}h(u)du\right).
\end{equation*}
Specification of the hazard function is therefore sufficient to specify the full likelihood model.

\subsection{Piecewise exponential models}
\label{sec;PEM}
Piecewise exponential models \citep{Feigl1965, Ibrahim2001} are constructed via a piecewise constant log-hazard function \eqref{eq;PEM_haz}. Covariates can be incorporated into \eqref{eq;PEM_haz} by replacing $\alpha_j$ with $\eta_{ij} = \alpha_j + w_i^\top\beta_j$. We refer to $\alpha_j$ as the local baseline log-hazard and $\beta_j \in \mathbb{R}^p$ as a vector of local covariate effects, which can encode a local proportional hazards assumption.

To complete the model specification we require priors for $\{\alpha_j, \beta_j, s_j\}$. Computational convenience is a common motivation for prior selection, primarily through the use of independent, conjugate Gamma priors on $\exp(\alpha_j)$.  Another common objective is some degree of smoothing between local hazards, by using either a random-walk prior on $\alpha_j$ \citep{Fahrmeir2001}, Markov-Poisson-Gamma priors \citep{Lin2021} or priors incorporating local and global trend terms \citep{Kearns2019}.   The independent priors are computationally convenient, but do not facilitate borrowing of information between segments. This limits the number of intervals that can be specified while leaving the model identifiable, resulting in a less expressive hazard function. Further, in the context of extrapolation, $\alpha_J$ is taken to be the hazard for the entire extrapolation period \citep{Cooney2023a}. This approach has been more broadly advocated for by \cite{Bagust2014} but has proved controversial due to the reliance on the model selected for the observation period, and the inability to inform long-term hazards with prior information \citep{Latimer2014}.  Smoothing priors borrow information between segments allowing for a more expressive hazard function in the observation period. The choice of prior for the observation period, however, is closely linked to the model for extrapolation, and either restricts the type or amount of prior information that can be incorporated or requires assumptions on the evolution of the hazard function in the observation period \citep{Kearns2022}. The prior we introduce in the following section will contain most of these prior structures as special cases, while providing a weakly informative prior during the observation period.

\subsection{Discretised Diffusion Priors}
\label{sec;diff_prior}
To capture prior knowledge about the long-term behaviour of the hazard we assume that the \textit{discrete}-time log-hazard process $\{\alpha_j\}_{j=1}^J$ can be described via a \textit{continuous}-time stochastic process $(\check{\alpha}_{\check{y}})_{\check{y}\geq 0}$ with dynamics governed by the stochastic differential equation
\begin{equation}
    \label{eq;SDE}
     \text{d}\check{\alpha}_{\check{y}} = \mu(\check{\alpha}_{\check{y}} )\text{d}\check{y} + \text{d}W_{\check{y}}, \quad \check{\alpha}_0 = a_0,
\end{equation}
with drift $\mu(\check{\alpha}_{\check{y}})$, where $(W_{\check{y}})_{\check{y}\geq 0}$ is a standard Brownian motion \citep{Oksendal2013}. The random variables $\alpha_1,...,\alpha_J$ are then defined through the relation $\alpha_j := \check{\alpha}_{j\sigma^2}$, where $\sigma^2$ is a step size defined later in this section. The primary motivation behind this prior is that information about the evolution of the hazard can be encoded into $\mu(\check{\alpha}_{\check{y}})$. During the observation period, where data are more abundant, this acts as a weakly informative prior with limited impact on the resulting inference. However, as observations become sparser and the hazard is extrapolated beyond $y_+$,  this prior naturally becomes increasingly informative, allowing for long-term inferences to be driven by expert opinion encoded through $\mu(\check{\alpha}_{\check{y}})$.

Previous works have utilised diffusions as priors for hazard functions, including \cite{Aalen2004} in which the hazard function is modelled as a squared Ornstein-Uhlenbeck process and \cite{Roberts2010} in which $\mu(\check{\alpha}_{\check{y}})$ is defined such that the resulting diffusion is a stochastic perturbation around a pre-specified hazard function. The challenges of working directly with diffusions are primarily computational. Diffusions of interest rarely have tractable solutions, and therefore need to be finely discretised, increasing computational cost. To combat this, our approach involves a hierarchical formulation in which the numerical discretisation is dictated by the knot locations $\{s_j\}_{j=1}^J$, which in turn are sampled from an underlying process, and a prior on the discretisation step-size. This allows for more parsimonious and computationally convenient hazard functions to be specified.  More details are given in Section \ref{sec;knot_prior}.

\subsubsection{Example choices of $\mu(\check{\alpha}_{\check{y}})$}

We briefly outline some example choices for $\mu(\check{\alpha}_{\check{y}})$, with their behaviour illustrated in Figure \ref{fig:Diffusions}. A first trivial example is to set $\mu(\check{\alpha}_{\check{y}}) = 0$. The underlying diffusion is then a Brownian motion and the discretised version recovers the random walk prior \citep{Fahrmeir2001}. In practice this corresponds to having no expert opinion about the long-term behaviour of the hazard, with credible intervals for the log-hazard increasing in width at a constant rate as $y \to y_\infty$. This assumption will often contradict available prior information, however, and can be improved upon in the following examples.

\textbf{Stationary distributions:} There will often be prior information available about a range of plausible values for the hazard function in the extrapolation period. In our framework this is encoded as a Langevin diffusion, such that $\mu(\check{\alpha}_{\check{y}}) = \nabla\log f_\psi(\check{\alpha}_{\check{y}})/2$, where $f_\psi(\check{\alpha}_{\check{y}})$ is the density of the required stationary distribution for the log-hazard, with parameters $\psi$. We consider log-Normal and Gamma (equivalently Normal and log-Gamma) stationary distributions for the hazard function (equivalently log-hazard function). The required drifts are then given by
\begin{equation}
    \label{eq;Langevin}
    \mu_{LN}(\check{\alpha}_{\check{y}}) = \frac{1}{\psi_2}(\check{\alpha}_{\check{y}} - \psi_1), \quad \mu_G(\check{\alpha}_{\check{y}}) = \psi_1 - \psi_2\exp(\check{\alpha}_{\check{y}}).
\end{equation}

\textbf{Underlying hazards:} The stochastic perturbation approach introduced by \cite{Roberts2010} can also be incorporated into our framework. In short we suppose that we have access to a known hazard function $h_0(y)$ that quantifies our belief about how the hazard function evolves in the extrapolation period derived, for example, from data from previous clinical trials. A suitable drift function can then be derived by viewing $h_0(y)$ as the solution to an autonomous ordinary differential equation,
\begin{equation*}
    \frac{dh_0(y)}{dy} = g(h_0(y)), \quad \mu_0(\check{\alpha}_{\check{y}}) = g(\check{\alpha}_{\check{y}}).
\end{equation*}
In \cite{Roberts2010}, the absolute value of the diffusion is used to map the diffusion from $\mathbb{R}$ to $\mathbb{R}_{>0}$. In our case $g(\check{\alpha}_{\check{y}})$ requires a final change of variables to be transformed to a drift for the log-hazard.  As a running example, we consider the case where $h_0(y)$ corresponds to a Gompertz hazard function. This is a natural choice as the Gompertz distribution is often used to model long-term survival in the general population \citep{Thatcher1999}, and therefore intuitively should provide sensible inferences for the extrapolation period. In our framework this ensures estimates of mean survival in the population of interest are consistent with those seen in the (usually healthier) general population. The corresponding stochastic differential equation has a linear drift
\begin{equation}
    \label{eq;Gompertz}
    \mu(\check{\alpha}_{\check{y}}) = \psi,
\end{equation}
where $\psi$  is the scale parameter of the required Gompertz distribution. The derivation of this quantity is provided in the supplementary material \citep{Hardcastle2024supp}.

\textbf{Time-varying drifts:} The above examples have utilised time-homogeneous drift functions. This is, however, not a necessary requirement. In particular, expert opinion on the evolution of the hazard function will often evolve with time. A more flexible class of diffusions can therefore be defined with time-varying drifts $\mu(\check{\alpha}_{\check{y}},y)$. We investigate this possibility further in Section \ref{sec;TimeVarying}.

\begin{figure}
    \centering
    \includegraphics[width=\linewidth]{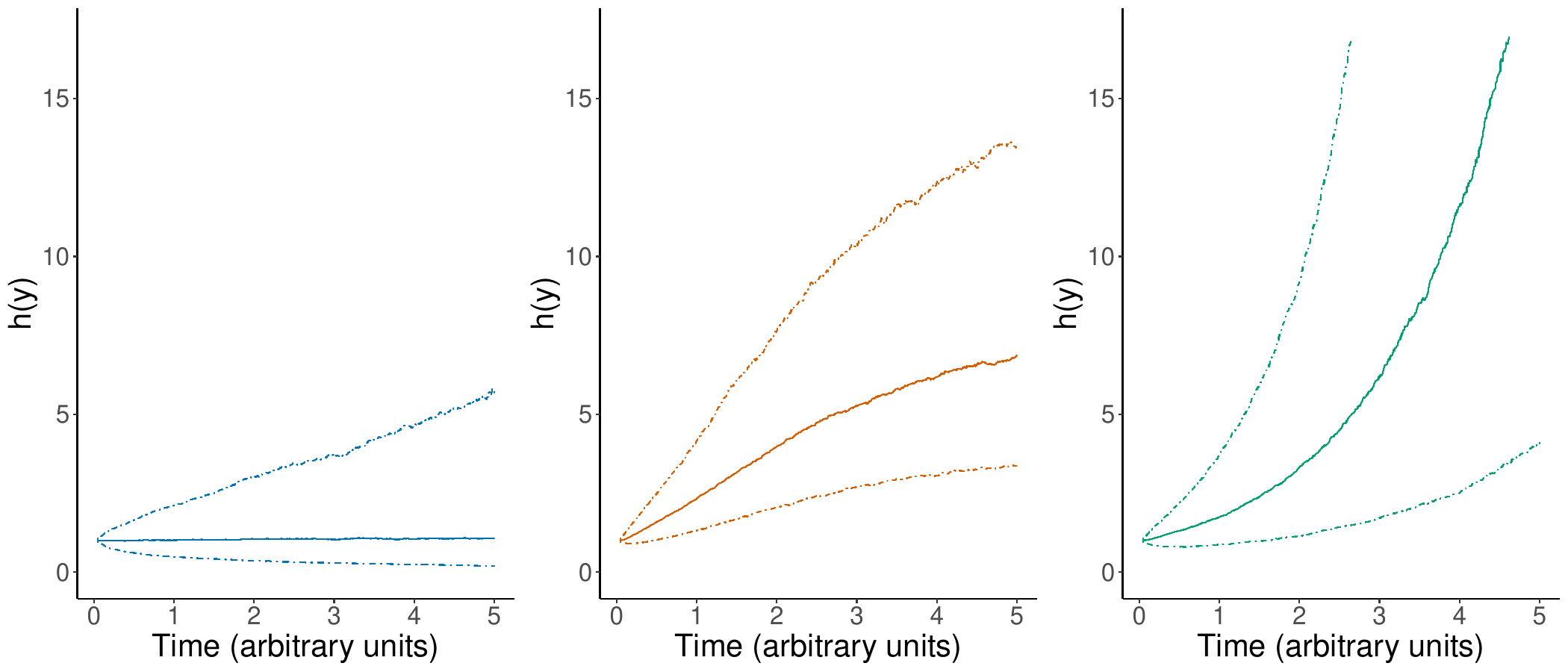}
    \caption{Prior simulations for $h(y)$ under different specifications for $\mu(\check{\alpha}_{\check{y}})$. (Left) Random Walk prior $\mu(\check{\alpha}_{\check{y}}) = 0$. (Centre) Gaussian Langevin prior \eqref{eq;Langevin}. (Right) Gompertz prior dynamics (log-linear drift) \eqref{eq;Gompertz}.}
    \label{fig:Diffusions}
\end{figure}

\subsubsection{Discretisation}

As noted previously, stochastic differential equations rarely have analytic solutions and therefore implementation requires \eqref{eq;SDE} to be discretised. The standard approach is the Euler-Maruyama discretisation \citep{Platen2010}
\begin{equation}
    \label{eq;EM}
    \check{\alpha}_{(j+1)\sigma^2} = \check{\alpha}_{j\sigma^2} + \theta_{j+1}, \quad \theta_{j+1} \sim \text{Normal}(\sigma^2\mu(\check{\alpha}_{j\sigma^2}),\sigma^2).
\end{equation}
where $\sigma^2$ is the step size, for $1 \leq j \leq J$. Note the slight abuse of notation, with $\{\check{\alpha}_{j\sigma^2}\}_{j=1}^J$ now used to denote the discretised version of $(\check{\alpha}_{\check{y}})_{{\check{y}}\geq0}$. It is well established that \eqref{eq;EM} can be numerically unstable when $\mu(\check{\alpha}_j)$ is not globally Lipschitz \citep{Roberts1996}. In our application this condition is particularly restrictive and is not satisfied, for example, by the log-Gamma Langevin drift \eqref{eq;Langevin}. More broadly, it is unrealistic to ask practitioners without a mathematical background to carefully check whether the drifts they elicit meet this condition before implementation, and an ideal generalisable prior would not rely on a Lipschitz drift.

To mitigate instabilities when considering non-Lipschitz drifts we utilise a recently introduced scheme based on skew-symmetric innovation densities \citep{Livingstone2024},
\begin{equation}
\label{eq;Barker}
    \check{\alpha}_{(j+1)\sigma^2} = \check{\alpha}_{j\sigma^2} + \theta_{j+1}, \quad f_0(\theta_{j+1}\mid \check{\alpha}_j) \propto \left(1 + \tanh(\mu(\check{\alpha}_{j\sigma^2})\theta_{j+1})\right)\phi(\theta_{j+1}\mid \sigma^2).
\end{equation}
Here $\phi(\cdot\mid\sigma^2)$ is the density of a $\text{Normal}(0,\sigma^2)$ random variable, and $1 + \tanh(\mu(\check{\alpha}_{j\sigma^2})\theta_{j+1})$ is a skewing term corresponding to the cumulative distribution function of a logistic distribution evaluated at $\mu(\check{\alpha}_{j\sigma^2})\theta_{j+1}$.\footnote{In fact, this construction is more general in that any CDF of a centred symmetric random variable is sufficient.} Similarly to the Euler-Maruyama method this approach introduces  approximation error that vanishes as $\sigma \to 0$.

Intuitively, while the Euler-Maruyama method \textit{shifts} $\theta_{j+1}$ in the direction of the drift, the skew-symmetric scheme \textit{skews} $\theta_{j+1}$ in the direction of the drift. This difference is depicted in Figure \ref{fig:Discretisation} for fixed $\sigma$ and increasing values of $\mu(\check{\alpha}_{\check{y}})$. In \cite{Livingstone2024} the authors show that \eqref{eq;Barker} is more robust than \eqref{eq;EM}, both to the choice of $\sigma$ and to non-globally Lipschitz $\mu(\check{\alpha}_{\check{y}})$. In both of the above cases we initialise the process at $\check{\alpha}_0 \sim \text{Normal}(0,\sigma_0^2)$. In Section \ref{sec;Mixing} we will also show that this approach is computationally advantageous, when combined with the prior for $\{s_j\}_{j=1}^J$ introduced in the following Section.

\begin{figure}
    \centering
    \includegraphics[width=\linewidth]{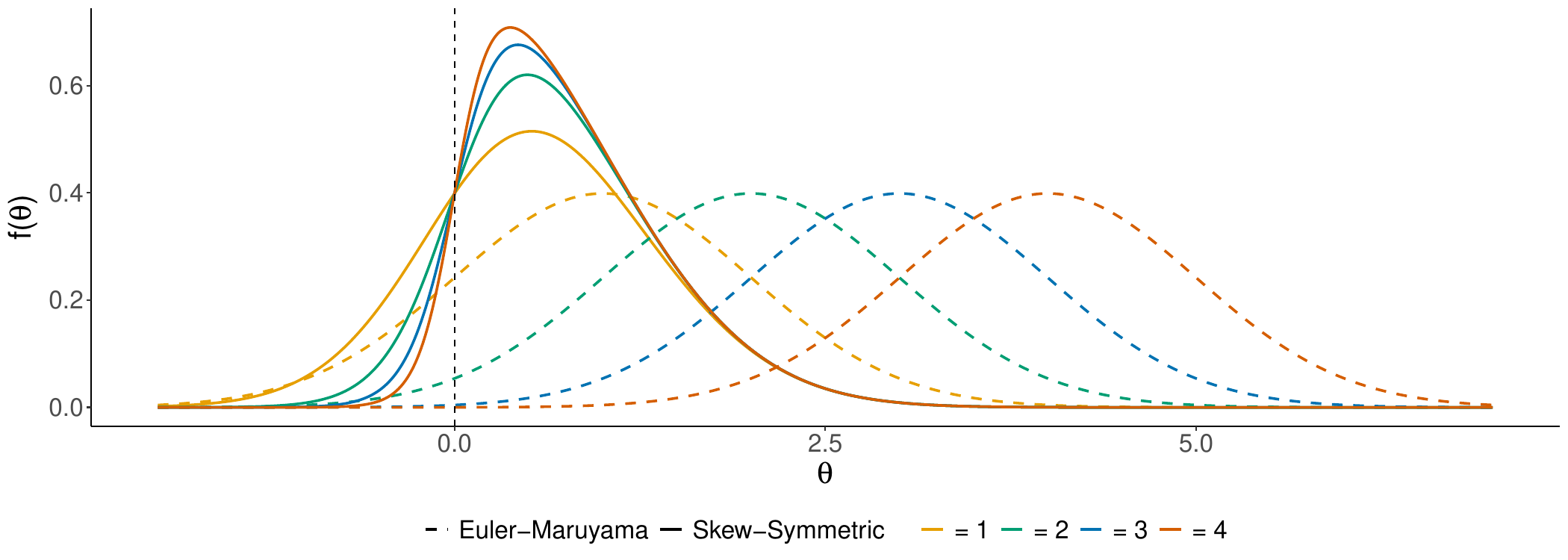}
    \caption{Density functions for the innovations $\theta$ under the Euler-Maruyama (dashed) and skew-symmetric (solid) schemes for increasing values of $\mu(\check{\alpha}_{\check{y}}) = 1, 2, 3, 4$ and fixed $\sigma^2$.}
    \label{fig:Discretisation}
\end{figure}

To complete the specification of the above process, we place an exponential prior on~$\sigma$,
\begin{equation*}
    \sigma \sim \text{Exponential}(a),
\end{equation*}
corresponding to a penalised-complexity prior \citep{Simpson2017}. This prior shrinks the innovation standard deviation towards 0, thus shrinking the overall hazard function towards a single constant value. In all the examples here we set the rate of the exponential prior to $a = 2$. Justification for this choice is provided in the supplementary materials \citep{Hardcastle2024supp}, however we expect inferences to be generally unaffected for sensible choices of $a$.

\subsection{A prior for knot locations}
\label{sec;knot_prior}
The specification of the model is completed with a prior for the knot locations, $\{s_j\}_{j=1}^J$. The standard approach is for $\{s_j\}_{j=1}^J$ to be fixed a priori, for example at set quantiles of observed event times or at regular intervals \citep{Murray2016}. The resulting hazard, however, will be sensitive to this specification, particularly in the absence of data during extrapolation period.

We address these issues directly by assuming that $\{s_j\}_{j=1}^J$ arise from a Poisson Point Process with intensity $\gamma$ on the interval $(0,y_+)$, denoted throughout as $\{s_j\}_{j=1}^J \sim \text{PPP}(\gamma, (0,y_+))$. This can be expressed equivalently as
\begin{equation}
    \label{eq;PPP}
    J \sim \text{Poisson}(y_+\gamma), \quad \{s_j\}_{j=1}^J \overset{iid}{\sim} \text{Uniform}(0,y_+).
\end{equation}
This prior (and variations) have been considered previously \citep{Chapple2020, Demarqui2012}; however, this specification is commonly avoided due to the computational challenges it introduces. 

Note that, in contrast to \citep[e.g]{Roberts2010}, in the above construction the discretisation step size $\sigma^2$ is independent of the distance between knots $(s_j - s_{j-1})$. Because of this, the prior for $\log h(y)$ is in fact given by \eqref{eq;SDE} through a random time-change defined by \eqref{eq;PPP}, such that a priori
\begin{gather*}
    \log h(y) = \alpha_{j} = \check{\alpha}_{j\sigma^2}, \quad
    j = \min\{l : y < s_l\}.
\end{gather*}
This construction can be viewed as first simulating a numerical skeleton $\{\check{\alpha}_{j\sigma^2}\}_{j=1}^J$, via \eqref{eq;Barker} and then mapping this to the time-scale of interest, $(0,y_\infty)$, via \eqref{eq;PPP}.

Under this prior the number of knots, and therefore the flexibility of the hazard function, is directly controlled by $\gamma$. The diffusion speeds up when the data require a more volatile hazard and slows down when the hazard is less volatile, adapting to the data without being constrained by the prior. In terms of extrapolation, the advantage of this formulation is that $\gamma$ determines the speed at which $\mu$ dominates the long-term hazard. Intuitively, if the hazard function is more volatile in the observation period we should expect the influence of the data to decay faster in the extrapolation period (with the prior taking over faster). Conversely, if the hazard is less volatile the data should remain informative for longer during the extrapolation period.

We consider two approaches for specifying $\gamma$. The first is to consider a set of models for a fixed number of values for $\gamma$. These models can then be compared using information criteria. A second, fully Bayesian approach places a prior $\gamma \sim \text{Gamma}(a,b)$, equivalent to a Negative Binomial prior on $J$. We compare these methods further in Section \ref{sec;Colon}.

\subsection{Incorporating covariates}
\label{sec;covariates}
We have focused so far on specification of a prior for the log-baseline hazard and corresponding knots. Both priors extend to the case when covariates are incorporated in the model. 

For the underlying diffusion, it suffices to provide a specification for each $\beta_j$ process, independently of the diffusion for $\alpha_j$. As $\beta_j$ is a covariate effect, a natural process to specify is a Langevin diffusion with Gaussian stationary distribution.\footnote{Note this is equivalent to specifying an Ornstein-Uhlenbeck prior for $\beta_j$.} Setting the mean to 0 implies the expected long-term treatment effect vanishes as $y\to\infty$. In Section \ref{sec;TimeVarying} we show that $\mu(\beta_j,y)$ can be modified to incorporate a waning long-term treatment effect, a common and important assumption in many HTA analyses \citep{Jackson2017}. Specifying a non-zero mean would imply a long-term proportional average treatment effect, but this would need to be supported by strong clinical opinion. Similarly for the prior for $\{s_j\}_{j=1}^J$, we define a set of knots $\{s_j^k\}$ independently of the set of knots for the baseline log-hazard. 

\section{Posterior sampling}
\label{sec;Sampling}
The diffusion piecewise exponential model generates several challenges for commonly used Bayesian inference engines primarily associated with the prior on $\{s_j\}_{j=1}^J$. The resulting posterior is transdimensional for which the standard sampling approach is to use reversible jump MCMC \citep{Green1995}. These samplers require the specification of a between-model proposal distribution that must be carefully tuned to achieve modest acceptance rates. This results in a noticeable increase in computational cost due to the additional likelihood evaluations required at each transdimensional step. 

Note, in addition to the above, that the fixed $\{s_j\}_{j=1}^J$ model can still present sampling challenges. The potential function, $U(x) :=-\log\pi(x)$ is non-Lipschitz, causing instability in gradient-based methods such as the Metropolis Adjusted Langevin Algorithm (MALA) \citep{Roberts1996} and Hamiltonian Monte Carlo (HMC) \citep{Livingstone2019}. Further, in the presence of high censoring rates (which is precisely the scenario we are considering), posteriors of survival models can exhibit high skew, again challenging standard Metropolis-Hastings methods \citep{Hird2020}.

To circumvent these issues we utilise sampling techniques based on continuous time Piecewise Deterministic Markov Processes (PDMPs) \citep{Fearnhead2024b}. These processes are non-reversible (e.g. \cite{Andrieu2021}) and use ballistic motion and gradient information to efficiently explore the target distribution. Further, in contrast to MALA and HMC, they have constant velocity and require minimal tuning making them more robust to non-Lipschitz potentials. Recent works have also shown that they are able to sample from transdimensional posteriors induced by spike and slab priors without the need for additional likelihood evaluations or tuning of between-model proposals \citep{Chevallier2023, Bierkens2023a}. The key contribution of this Section is to show how these results can be extended to more general transdimensional posteriors. A concise presentation of the algorithm is given in the supplementary material \citep{Hardcastle2024supp}.

We briefly note that for posterior sampling we use a non-centred parameterisation of the model \citep{Betancourt2015}
\begin{equation*}
    \Tilde{\theta}_0 = \alpha_0 \quad \Tilde{\theta}_j = \sigma^{-1}\theta_j.
\end{equation*}
This avoids strong posterior dependence between $\alpha_j$'s and eliminates funnel-shaped geometry that can arise when simultaneously updating $\theta$ and $\sigma$ (e.g. \cite{Betancourt2015}).

\subsection{The Bouncy Particle Sampler and Forward event chain Monte Carlo}
\label{sec;BPS}

Piecewise Deterministic Monte Carlo\footnote{In the Statistical Physics literature, where several of these ideas were originally developed, these methods are commonly referred to as Event Chain Monte Carlo.} methods have emerged as a promising class of non-reversible processes for posterior sampling in challenging Bayesian inference problems. In this work we use a variation of the bouncy particle sampler \citep{Bouchard2018}, known as Forward Event Chain Monte Carlo \citep{Michel2020}. A more general introduction to these methods can be found in \cite{Fearnhead2024b}. 

Given the current sampler time, $t$, the bouncy particle sampler is defined on an state-space augmented with velocities $z_t = (x_t, v_t) \in \mathbb{R}^d\times\mathbb{S}^{d-1}$, with $x = (\Tilde{\theta}, \sigma)$ and $d = J(p+1) + 1$. The continuous-time deterministic evolution of $z_t$ is given by the system of ordinary differential equations
\begin{equation*}
    \frac{dx_t}{dt} = v_t, \quad \frac{dv_t}{dt} = 0,
\end{equation*}
which results in $v_t$ driving the linear evolution of $x_t$. This evolution is interrupted by a jump process, with jump times given by the inhomogeneous event rate, $\lambda(t)$, and the transformation of $z_t$ at these times given by a  deterministic map $Q$. For the standard bouncy particle sampler these are defined as
\begin{gather*}
    \lambda(t) = \max\{0, \langle v_t, \nabla U(x_t)\rangle\} + \lambda_r, \\
    Q: (x_t, v_t) \mapsto (x_t, v_t - 2v_t^{\nabla U}).
\end{gather*}
Here, $\lambda_r \in \mathbb{R}_{\geq 0}$ is the refreshment rate, with $\lambda_r > 0$ required to ensure the process is irreducible, and $v_t^{\nabla U}$ arises from the orthogonal decomposition of $v_t$ with respect to $\nabla U(x_t)$, $v_t = v_t^{\nabla U} + v_t^{\perp}$, with $v_t^{\perp} \perp \nabla U(x_t)$. Events associated with the first term of $\lambda(t)$ only occur when $\langle v, \nabla U(x_t)\rangle\ > 0$, i.e when the process is moving into areas of lower posterior mass, resulting in fast convergence towards areas of high posterior mass, meanwhile the map $Q$ corresponds to a reflection of the velocity off the tangent to the potential. The above process has stationary distribution $\rho(z) \propto \pi(x)\nu(v)$, where $\nu(v)$ is the uniform measure on $\mathbb{S}^{d-1}$. Samples from $\pi(x)$ are then recovered by marginalising out $v$.

A well-known drawback of the bouncy particle sampler is that $\lambda_r$ requires careful tuning, and that optimal values of $\lambda_r$ can result in approximately $70\%$ of events being refreshments \citep{Bertazzi2022}. This replaces the ballistic motion of the process with increasingly diffusive dynamics, inhibiting sampling efficiency. To remedy this issue the forward event chain method \citep{Michel2020} replaces the the deterministic mapping $Q$ with a jump kernel stochastically updating both of $(v_t^{\nabla U}, v_t^{\perp})$. This incorporates refreshment into reflections while ensuring the process targets the correct stationary distribution. In particular, $v_t^{\perp}$ is re-sampled as $\Tilde{v}_t^{\perp}$ such that $\langle v_t^{\perp}, \Tilde{v}_t^{\perp} \rangle \geq 0$, reducing the diffusivity as compared to full refreshments. Further, the update to $v_t^{\perp}$ need not occur at every event, but can be set to update at the first event after each time given by a homogeneous Poisson process with rate $\lambda_e$. 

To understand the robustness these alterations introduce we can consider the process when $\lambda_e$ is poorly tuned. In the case $\lambda_e$ is set to be too large the refreshment rate is capped by the rate at which reflections occur. Conversely, when $\lambda_e$ is too small, the stochastic updates of $v_t^{\nabla U}$ ameliorate the irreducibility issues observed in the original bouncy particle sampler. In contrast, poor tuning of the refreshment rate $\lambda_r$ in the bouncy particle sampler can significantly impact the resulting process. The forward event chain approach has seen uptake in the statistical physics literature, but we believe this to be the first application to an applied Bayesian statistics problem. The specific strategies from \cite{Michel2020} used in this work are outlined in the supplementary material \citep{Hardcastle2024supp}.

\subsection{Generating the process} 
\label{sec;PDMP_gen}

The deterministic dynamics and jump kernel of the bouncy particle sampler and forward event chain Monte Carlo are simple to generate, but the inhomogeneous Poisson process associated with $\lambda(t)$ is typically more challenging, and an area of active research \citep{Andral2024, Corbella2022, Sutton2023}. The primary method we use to generate this event rate is the splitting schemes approach of \citep{Bertazzi2023}, which alternates between updating the deterministic and event rate processes over a given time step $\Delta t$. Without adjustment this scheme introduces a small approximation error into the posterior, which could in principle be corrected for using a non-reversible Metropolis--Hastings filter as described in \cite{Bertazzi2023}, though we found this to be unnecessary here. Note that this replaces the continuous time sample paths of the original process with a discrete time approximation.

A second exact scheme we consider is to update $\sigma$ using conditional Metropolis-within-Gibbs steps at exponential times in the sampler \citep{Sachs2023}. For the random walk, Gaussian Langevin and Gompertz drifts, the potential for $\Tilde{\theta}$ is then convex and the process can be generated exactly by determining event times using a line search \cite[Example 1]{Bouchard2018}. 

Of the two schemes we prefer the first. Both algorithms have tuning parameters that are easy to specify, although we find this is marginally easier to do in the former case. Further the sampling efficiency of the second method seems to be inhibited, both by reversible updates for $\sigma$, and conditional updating that struggles to explore the geometry of the posterior. Finally, the splitting method allows for a general drift function $\mu$ in the diffusion prior to be specified, which is an important goal of this work. Full algorithms for both methods are produced in the supplementary materials \citep{Hardcastle2024supp}.

\subsection{Spike and slab PDMPs}
\label{sec;SMPDMP}

Transdimensional posteriors are often induced by priors that are mixtures of continuous and atomic components (commonly referred to as spike and slab priors) of the form
\begin{equation}
    \label{eq;SnS}
    \pi_0(d\Tilde{\theta}_j) \propto (1-\omega)\delta_0(d\Tilde{\theta}_j) + \omega f_0(d\Tilde{\theta}_j),
\end{equation}
where $\delta_0$ is a Dirac mass at 0, $f_0$ is a continuous density and $\omega \in (0,1)$.  In all of our examples we set $\omega = 0.5$.

These posteriors can be sampled from directly using the forward event chain sampler, by moving from the continuous component to the atomic component at exactly the point when $\Tilde{\theta}$ intersects the hyperplane $\{\Tilde{\theta} : \Tilde{\theta}_j 
 = 0\}$ \citep{Bierkens2023a,Chevallier2023}. Equivalently this can be seen as setting $v_j \mapsto 0$ at this point, with an appropriate renormalisation step when $v \in \mathbb{S}^{d-1}$. For forward event chain Monte Carlo, $v_j$ is then refreshed after an exponential time, $\tau$ with 
 \begin{equation}
    \label{eq;stickyrate}
     \tau \sim \text{Exponential}\left(\frac{\omega}{1-\omega}f_0(0)|\mathcal{J}|\right),
 \end{equation}
 where $|\mathcal{J}|$ is the Jacobian associated with renormalising $v$. The remaining terms in this rate are given by a posterior ratio, between the model where $\Tilde{\theta}_j$ is on the slab and $\Tilde{\theta}_j$ is on the spike. Homogeneity of \eqref{eq;stickyrate} arises due to the transdimensional updates occurring at a point where the likelihoods in both models are equivalent and therefore cancel, along with the majority of prior terms, simplifying this posterior ratio. When multiple components are considered simultaneously the next unsticking time is simply given by summing together the unsticking rates, with the component to update then selected uniformly at random.
 
The construction of \eqref{eq;stickyrate} is remarkable as, in contrast to most reversible jump MCMC methods \citep{Green1995}, transdimensional updates do not require either the specification of tuning parameters or likelihood evaluations. Following the terminology of \cite{Bierkens2023a} we will refer to these dynamics as sticky PDMP dynamics from this point forward.

\subsection{Sticky PDMPs for knot selection}
\label{sec;PDMP_knot}

While the computational efficiency and lack of tuning parameters in the above construction is appealing they have not yet been applied to transdimensional posteriors beyond those induced by spike and slab priors. We now show that the above dynamics can be extended to sampling the location of knots under a Poisson process prior following a two step procedure: \textit{i)} Given a fixed set of candidate knots, use sticky PDMP dynamics to update which knots are active in the model and which are inactive. \textit{ii)} Use a Gibbs step to update the set of candidate knots solely through updating the set of inactive knots. 

\subsubsection{Updating given fixed candidate knot locations}

We begin by considering the simpler case in which a fixed a set of unique candidate knot locations $\{m_i\}_{i=1}^M$ with scaled innovation parameters $\Tilde{\theta} \in \mathbb{R}^M$ are chosen. We will assume that this set is composed of a set of active knots $\{s_j\}_{j=1}^J$ such that $m_i \in \{s_j\}_{j=1}^J$ implies that $\Tilde{\theta}_i \neq 0$ almost surely, and a set of inactive knots, $\{r_j\}_{j=1}^{M-J}$, such that $m_i \in \{r_j\}_{j=1}^{M-J}$ implies $\Tilde{\theta}_i = 0$. Assuming a priori that $\mathbb{P}(m_i \in \{s_j\}_{j=1}^J) = \omega$ is equivalent to defining a spike and slab prior introduced in \eqref{eq;SnS} independently for each $\Tilde{\theta}_i$, with $f_0(d\Tilde{\theta}_i)$ corresponding to prior density for $\Tilde{\theta}_i$ induced by \eqref{eq;Barker}.

Sticky PDMP dynamics can then be directly applied without modification, with moves onto and off the spike updating membership of $\{s_j\}_{j=1}^J$ and $\{r_j\}_{j=1}^{M-J}$. When viewed in $\alpha$-space, the resulting dynamics split and merge the trajectory of neighbouring $\alpha$'s in continuous time, showcasing a natural connection to split-merge reversible jump moves used in several settings \citep{Brooks2003}. These dynamics are illustrated in Figure \ref{fig:SMPDMP}. 

Note, these dynamics cannot be immediately extended to the model with a Poisson process prior, as the set of candidate knots is uncountable. Each element of $\{r_j\}$ has an associated Poisson clock with rate defined in \eqref{eq;stickyrate}, and therefore the resulting combined unsticking rate will be infinite unless only a countable number of them are non-zero.

\begin{figure}
    \centering
    \includegraphics[width=\linewidth]{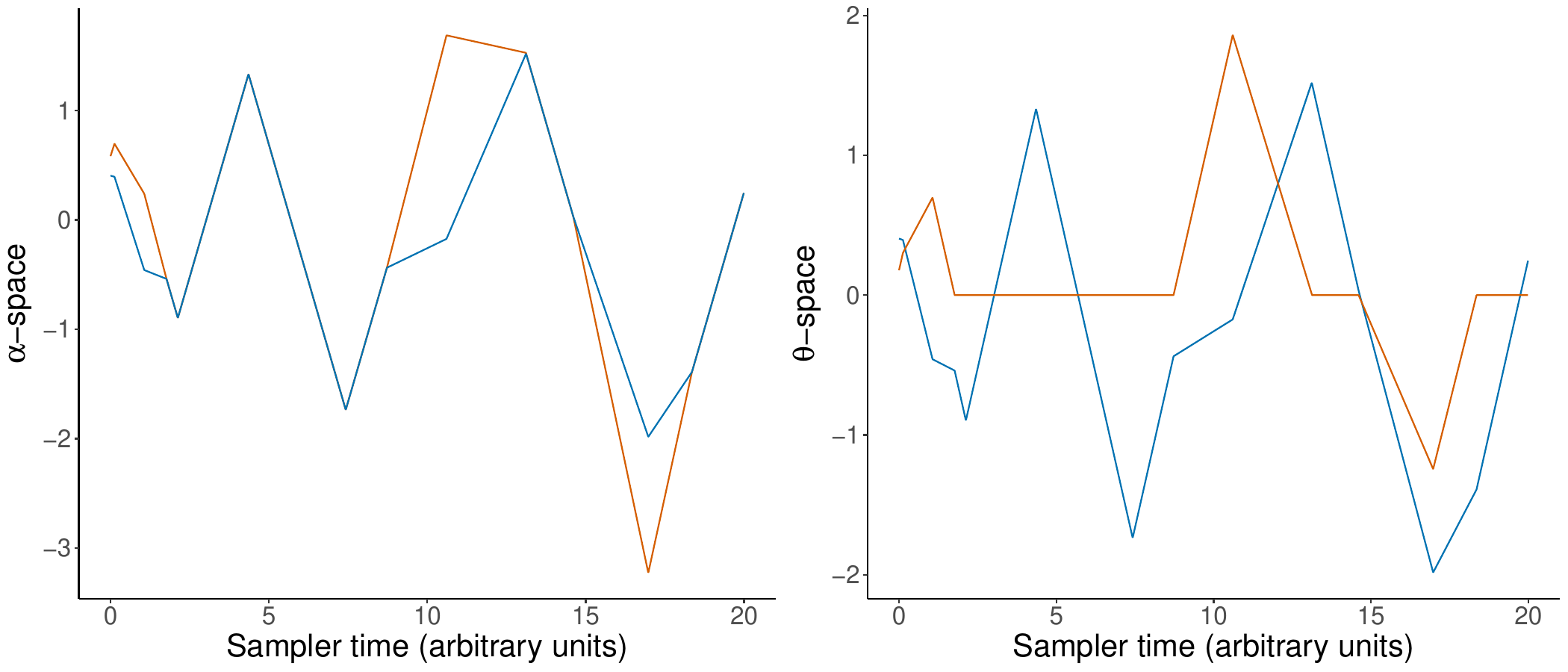}
    \caption{Trajectories for the PDMP sampler for knot selection viewed on $\alpha$-space (left) and $\theta$-space (right).}
    \label{fig:SMPDMP}
\end{figure}

\subsubsection{Updating the set of candidate knots}

The second part of this procedure circumvents this explosivity, by initialising the sampler with a finite set of candidate knots that is then regularly updated via a Gibbs step. To define this update, we first let the intensity $\gamma := \omega\Gamma$ with $\omega$ defined in equation \eqref{eq;SnS} and $\Gamma > 0$. This does not alter the prior introduced in Section \ref{sec;knot_prior}.  Under this specification the set of (now random) candidate knot locations $\{m_i\}_{i=1}^M \sim \text{PPP}(\Gamma, (0,y_+))$, and $\{s_j\}_{j=1}^J$ can be viewed as a thinned version of this process with thinning probability $\omega$. As in the previous section, this is equivalent to defining a spike and slab prior \eqref{eq;SnS} independently for each $\Tilde{\theta}_i$.

A valid and computationally efficient Gibbs step then proceeds by re-sampling $\{r_j\}_{j=1}^{M-J} \sim \text{PPP}((1-\omega)\Gamma, (0,y_+))$. As these knots are inactive, updating their location does not alter the value of the likelihood and they can therefore be drawn directly from the prior. Conversely, if $\{m_i\}_{i=1}^M$ was updated this would require a Metropolis correction with corresponding likelihood evaluations. These Gibbs updates can occur at fixed or exponentially distributed times in the sampler \citep{Sachs2023}. Further, if a hyperprior has been placed on $\gamma$, this can also be updated at these times. These steps are shown in full in the algorithms presented in the supplementary materials \citep{Hardcastle2024supp}.

\subsection{Mixing time of the process}
\label{sec;Mixing}
The efficiency of the above process is dependent on the value of $f_0(0)$, i.e the continuous part of the prior for $\Tilde{\theta}_j$ evaluated at 0. This can be seen through \eqref{eq;stickyrate}, as smaller values of $f_0(0)$ result in longer sticking times at 0, requiring the process to be run for longer to obtain the same estimates. We can in fact formalise this intuition to compare the mixing times of the process under different parameterisations of the underlying diffusion.
\begin{proposition}
    \label{prop;Mixing} 
    Given a fixed set of candidate knots, let $\tau_0^{S}$ (respectively $\tau_0^{EM}$) be the recurrence time to the null model (i.e the model when all knots are inactive) under the skew-symmetric parameterisation (respectively the Euler-Maruyama parameterisation). Then
    \begin{equation}
        \label{eq;Prop3.1}
        \mathbb{E}[\tau_0^{S}] \leq \mathbb{E}[\tau_0^{EM}].
    \end{equation}
\end{proposition}
\begin{proof}
    Following \cite[Remark 2.4]{Bierkens2023a}, as the process is invariant the expected recurrence time to the null model is inversely proportional to the expected occupation time in the null model,
    \begin{equation}
        \label{eq;Mixing}
        \mathbb{E}[\tau_0] \propto \left(\frac{\omega}{1-\omega}f_0(0\mid \Tilde{\theta},\sigma)|\mathcal{J}|\right)^{-M} \propto f_0(0\mid \Tilde{\theta},\sigma)^{-M}.
    \end{equation}
    Then note that under the skew-symmetric parameterisation, $f^{S}_0(0\mid \Tilde{\theta},\sigma)$ is the density of a standard Normal distribution evaluated at 0, as the skewing term equals one when $\Tilde{\theta}_j = 0$. Further under the Euler-Maruyama scheme the density, $f^{EM}_0(\cdot\mid \Tilde{\theta},\sigma)$ is that of a $\text{Normal}(\sigma^2\mu(\alpha_{j-1}),1)$ distribution. Therefore $f^{EM}_0(0\mid \Tilde{\theta},\sigma) \leq f^{S}_0(0\mid \Tilde{\theta},\sigma)$ and \eqref{eq;Prop3.1} follows directly.
\end{proof}

A direct consequence of Proposition \ref{prop;Mixing} is that we can expect faster mixing times under the skew-symmetric parameterisation. We support this argument empirically by examining the performance of the sampler under each parameterisation for identical $\mu(\alpha_j)$. In particular we consider mean 0 Gaussian Langevin diffusions with standard deviations $\phi_2 = 2$ and $\phi_2 = 0.2$, for increasing values of fixed $\sigma$. Note that the resulting $\mu(\alpha_j)$ is Lipschitz and the approximation of the drift should be stable under both parameterisations. 

Figure \ref{fig:param_exp} shows the resulting estimates of $\omega$. For the case $\phi_2=2$, $\mu(\alpha)$ is relatively flat and so both parameterisations provide good sampling for small values of $\sigma$, however the Euler-Maruyama parameterisation becomes increasingly unstable as $\sigma$ increases. For $\phi_2=0.2$ the Euler-Maruyama parameterisation remains stuck either on or off indicating noticeably slower mixing. There is larger variance in the estimates provided by the skew-symmetric parameterisation for larger values of $\sigma$, but the sampling is clearly improved compared to the Euler-Maruyama parameterisation. 

\begin{figure}
    \centering
    \includegraphics[width=\linewidth]{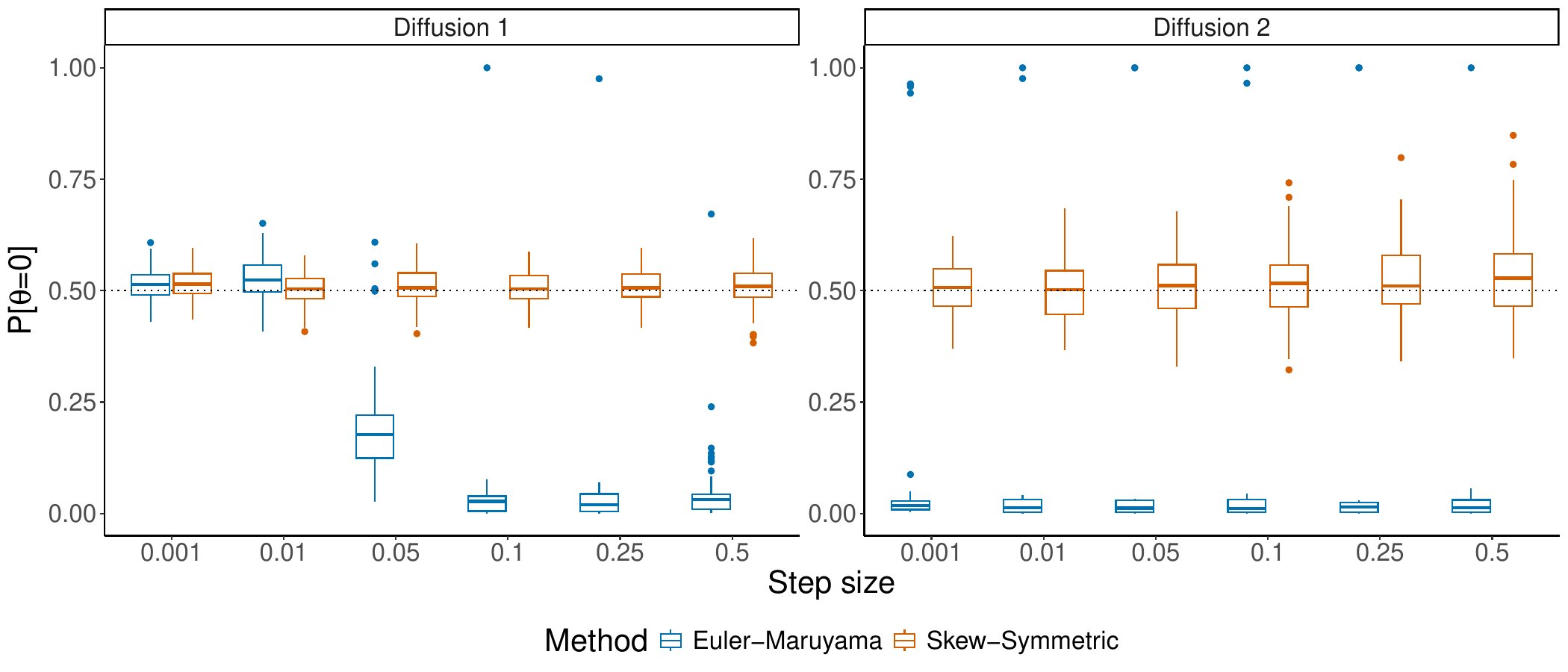}
    \caption{Comparison of efficiency of the PDMP sampler under the skew-symmetric and Euler-Maruyama parameterisations for different fixed values of $\sigma$. (Diffusion 1) $\mu(\alpha_j) = \alpha_j/2^2$, (Diffusion 2) $\mu(\alpha_j) = \alpha_j/0.2^2$. The dotted line indicates the true value $\omega = \mathbb{P}(\theta = 0) = 0.5$. }
    \label{fig:param_exp}
\end{figure}

\subsection{Generating extrapolations}
\label{sec;Extrap}
We note that the sampling methodology presented in this section has been designed to sample from parameters corresponding to the observed data period. Sampling parameters for the extrapolation period is easily handled using the skew-symmetric scheme directly along with posterior samples for $(\alpha_{M}, \sigma, \gamma)$. This direct sampling is more efficient than using the PDMP in the absence of data, and helps mitigates strong posterior dependencies that arise over extended time horizons. 

Discretising the diffusion does introduce a first order bias that vanishes as $\sigma \to 0$. To reduce the bias in the extrapolation period, the $(\gamma, \sigma)$ can be rescaled during this procedure. Full details are provided in the supplementary materials \citep{Hardcastle2024supp}.

\subsection{Comparison to reversible jump}
\label{sec;rj_comp}

To understand the efficiency of the developed methodology we compare the sampler to a comparable reversible jump scheme consisting of alternating an update for $\{s_j\}_{j=1}^J$ by either adding or removing a knot at each iteration with a random walk Metropolis update for $\tilde{\theta}$. We prefer the Random Walk to other choices of proposal kernel due to its robustness to tuning parameters that can be challenging to tune correctly within transdimensional sampling algorithms \citep{Livingstone2022}. 

We fit the diffusion piecewise exponential model to the Colon data set analysed in Section \ref{sec;Colon} using both the introduced PDMP sampler and the reversible jump sampler run for the same computational budget. Plots of the resulting hazard functions are shown in Figure \ref{fig:rj_exp}, along with trace plots for $h(1.2)$. Notably, the reversible jump sampler is unable to sufficiently explore the tails of the posterior for the hazard function, and therefore underestimates posterior uncertainty.

\begin{figure}
    \centering
    \includegraphics[width=\linewidth]{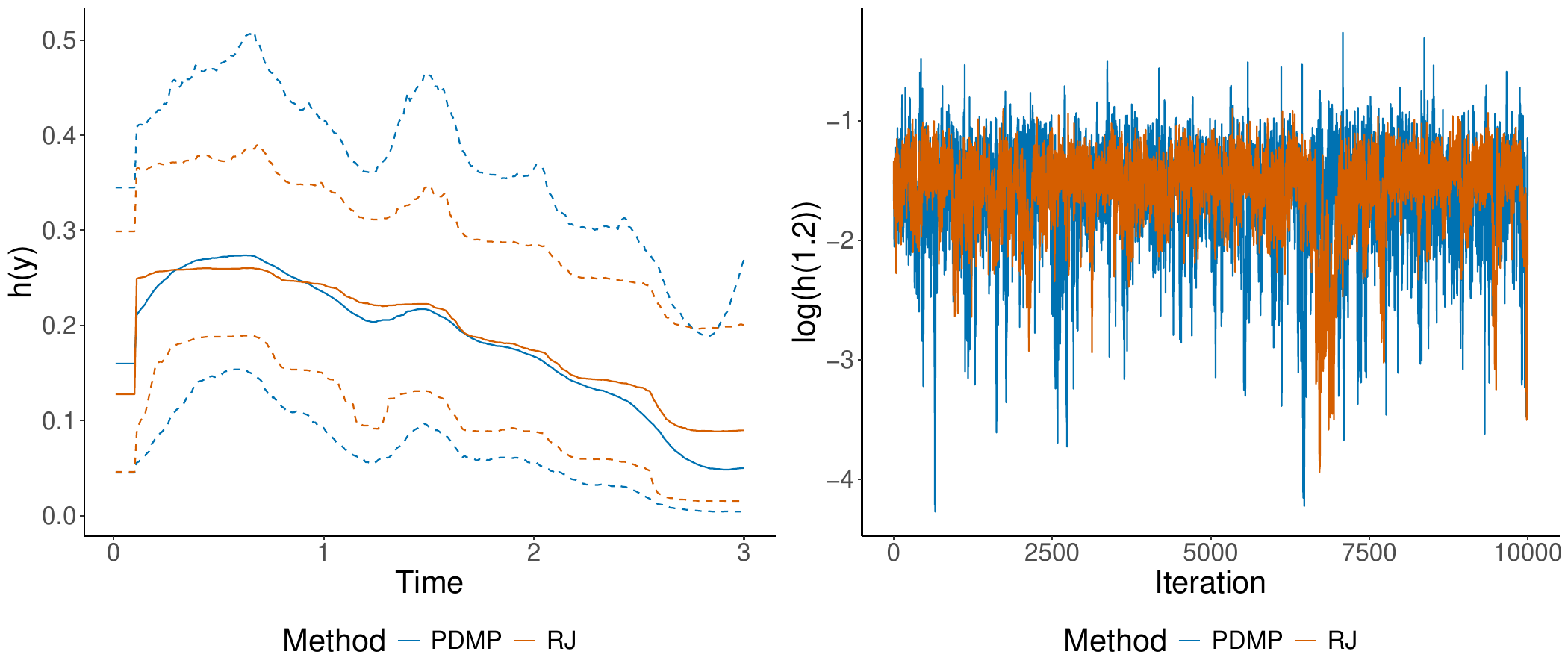}
    \caption{(Left) Inferred hazards under the reversible jump sampler (Orange) and the PDMP sampler (blue) with median hazards (solid) and 95\% credible intervals (dashed) reported. (Right) Trace plots for $\log(h(1.2))$ under the two samplers. Note that the reversible jump sampler is struggling to fully explore the tails of the hazard function.}
    \label{fig:rj_exp}
\end{figure}

It is natural to wonder how the design choices we have made affect the efficiency of the reversible jump sampler. In the supplementary materials we show results for alternative values of tuning parameter in the reversible jump proposal and full details of the algorithm \citep{Hardcastle2024supp}. Further, we provide a comparison to the sampler introduced by \cite{Chapple2020} for a similar model specification. 

\section{Applications}
\label{sec;Examples}
\subsection{Colon Cancer data}
\label{sec;Colon}
Our first illustrative application is to a dataset consisting of survival times from 191 colon cancer patients, of whom 22 were censored before 3 years and 104 were administratively censored at 3 years. This data is available via the \texttt{survextrap} R package \citep{Jackson2023}. To implement the model the practitioner is required to specify two quantities, the hyperprior (or fixed value) for $\gamma$, and the drift $\mu(\alpha_j)$.

\subsubsection{Specifying $\gamma$}

We consider both methods for the specification of $\gamma$ highlighted in Section \ref{sec;diff_prior}, namely \textit{i)} Selecting an optimal value of $\gamma$ based on information criteria. \textit{ii)} Placing a hyperprior on $\gamma$. 

Information criteria are commonly used when selecting a model for survival extrapolation \citep{Baio2020}. These results must be combined with an assessment of the plausibility of extrapolated hazards, however, as information criteria only assess goodness-of-fit within the observation period, providing no guarantees for the quality of extrapolations. As a result, analysts are often faced with the choice of either selecting a model that fits the observed data poorly or a model that exhibits unrealistic long-term behaviour. We note that the use of a more flexible model is not an automatic remedy to this issue. If no additional information is provided to guarantee the quality of extrapolations, then the above scenario will always be a possibility.

The diffusion piecewise exponential model avoids this trade-off through the specification of $\mu$, breaking the dependence between the model fit to the observation period and the limiting behaviour of the hazard function. As such information criteria can be used to select $\gamma$. The practical impact of this choice beyond $y_+$, is to control the rate at which the influence of the data in the observation period decay. Intuitively, if the observed hazard is more volatile, we can expect this influence to decay faster in comparison to a more stable hazard function.

In this work we use the leave-one-out information criteria estimated using Pareto-smoothed importance sampling \citep{Vehtari2017} due to its stability properties compared to alternative criteria. The same criteria are used in \citep{Jackson2023} to determine the number and location of knots when using M-splines. We believe that our approach is simpler, however, as it requires only the selection of a single parameter.

We find the approximation to the leave-one-out cross-validation score does not always sufficiently penalise overly complex models, resulting in implausibly shaped hazard functions. We therefore suggest that this score should be minimised, while also ensuring the shape of the hazard function remains plausible. The optimal value for the colon cancer data, using $\mu(\alpha_j) = 0$, is $\gamma = 7$. The full results of this procedure are available in the supplementary materials \citep{Hardcastle2024supp}.

For the second approach, to allow for consistent comparisons with the above procedure we specify 
\begin{equation*}
    \gamma \sim \text{Gamma}(7,1).
\end{equation*}
As noted previously we can view this as a Negative Binomial prior. Previous applications of Negative Binomial priors in similar contexts have found they are less informative than Poisson priors for $J$ \citep{Sharef2010}. In practice we find that while the Negative Binomial prior is robust to the specification of the overdispersion parameter, in the sense that posterior inferences are minimally affected, the specified prior mean can still be influential. Therefore in practice modelling needs to be coupled with sensitivity analysis to understand the influence of this choice.

\subsubsection{Specifying $\mu(\alpha_j)$}

Specification of $\mu(\alpha_j)$ drives the behaviour of the hazard function during the extrapolation period, and should be elicited using expert opinion or external data on the long-term behaviour of the hazard. In particular it should \textit{not} be selected using information criteria, as this only measures predictive ability during the observation period. 

We consider various specifications of the time-homogeneous drifts outlined in Section \ref{sec;diff_prior}. The use of Langevin diffusions with log-Gamma or Gaussian stationary distributions encodes an assumption that the expected hazard function will be constant as $y\to\infty$. To illustrate the method in the following examples we use Langevin diffusions with $\text{Normal}(\log(0.29),0.4)$ and $\text{log-Gamma}(2,7)$ stationary distributions for the log-hazard, and the Gompertz diffusion \eqref{eq;Gompertz} with $\psi = 0.3$. For each model generating two chains of 10,000 samples including burn-in took approximately 45 seconds. Examples of how to derive these prior drifts, full computational and modelling details are provided in the supplementary material \citep{Hardcastle2024supp}.

\subsubsection{Results} 

Mean survival estimates for each specification of $\mu(\alpha_j)$ under the Poisson and Negative Binomial priors are presented in Table \ref{tab:Colon}, with corresponding hazard functions in Figure \ref{fig:Colon_haz}. Posterior mean survival estimates for the observation period are almost identical under each specification of $\gamma$ and $\mu(\alpha_j)$, with inferences driven by the observed data. Similarly, $\mu(\alpha_j)$ has minimal influence on the hazard functions in the observation period, although notably the Negative Binomial prior provides a smoother fit than the corresponding Poisson prior. 

\begin{table}[]
    \centering
    \begin{tabular}{l|l|l}
       Model   & $\mathbb{E}[Y]$ on $(0,y_+)$ & $\mathbb{E}[Y]$ on $(0,y_\infty)$ \\
       \hline
        Random Walk (Poisson)     & 2.19 (2.01, 2.36) & 4.73 (3.14, 6.09) \\
        Random Walk (Neg. Binomial)    & 2.21 (2.02, 2.38) & 4.67 (3.21, 6.06) \\
        Log-Normal stationary (Poisson) & 2.19 (1.99, 2.36) &  3.80 (3.22,  4.49) \\
        Log-Normal stationary (Neg. Binomial) & 2.20 (1.99, 2.39) & 3.97 (3.26, 4.82) \\
        Gamma stationary (Poisson)         & 2.19 (2.01, 2.36) & 4.31 (3.29, 5.52) \\
        Gamma stationary (Neg. Binomial)   & 2.21 (2.01, 2.39) & 4.39 (3.34, 5.57) \\
        Gompertz (Poisson)         & 2.19 (2.02, 2.36) & 4.43 (2.94, 5.89) \\
        Gompertz (Neg. Binomial)   & 2.20 (2.01, 2.38)  &  4.44 (3.03, 5.84) \\
        \hline
        Log-normal parametric   &  2.18 (2.03, 2.32) & 5.79 (4.71, 6.86) \\
        Independent piecewise exponential     & 2.27 (2.11, 2.42) & 5.37 (4.10, 7.34) \\
        M-spline (final knot = 5) & 2.25  (2.10, 2.40) & 6.89 (4.65, 9.10)\\
        M-spline (final knot = 10) & 2.25 (2.10, 2.41) &  6.57 (4.00, 8.81) \\
        M-spline (final knot = 15) & 2.26 (2.10, 2.40)  & 6.45 (3.65, 8.56) \\
        \hline
    \end{tabular}
    \caption{Mean survival estimates for the colon cancer data for the observation period and total window of interest with 95\% credible intervals. (Top) Estimates under varying specifications of $\mu(\alpha_j)$ for both the Poisson and Negative Binomial priors. (Bottom) Estimates from the log-normal standard parametric model, an independent piecewise exponential model, and M-spline hazard model.}
    \label{tab:Colon}
\end{table}

In contrast, mean survival estimates in the extrapolation period are highly reliant on the information encoded in $\mu(\alpha_j)$. In particular the credible intervals under the Random Walk and Gompertz drifts are larger than those under the Log-Normal and Gamma Langevin drifts. This difference in behaviour can also been seen in the hazard functions, where the credible intervals are noticeably larger under the former prior specifications. In general, although not for the random walk prior, the Negative Binomial specification results in larger estimates of mean survival. This is due to the smoother hazard function inferred for the observation period, slowing the speed of the underlying diffusion and the corresponding rate that the influence of the prior grows. Note that this behaviour is because the prior information encodes a typically higher hazard value than that observed at the end of the observation period. If the converse were true, then the Negative Binomial prior would result in more conservative estimates of mean survival. Finally, Figure \ref{fig:Colon_haz} shows that the Gompertz drift results in large credible intervals (larger in fact than the random walk prior), suggesting this prior does not encode much information in the extrapolation period. This is due to the exponential form of the Gompertz hazard function. As such, extrapolations are highly sensitive to the hazard observed at the end of observation period. We explore improvements to this specification in Section \ref{sec;TimeVarying}.

\begin{figure}[]
    \centering
    \includegraphics[width=\linewidth]{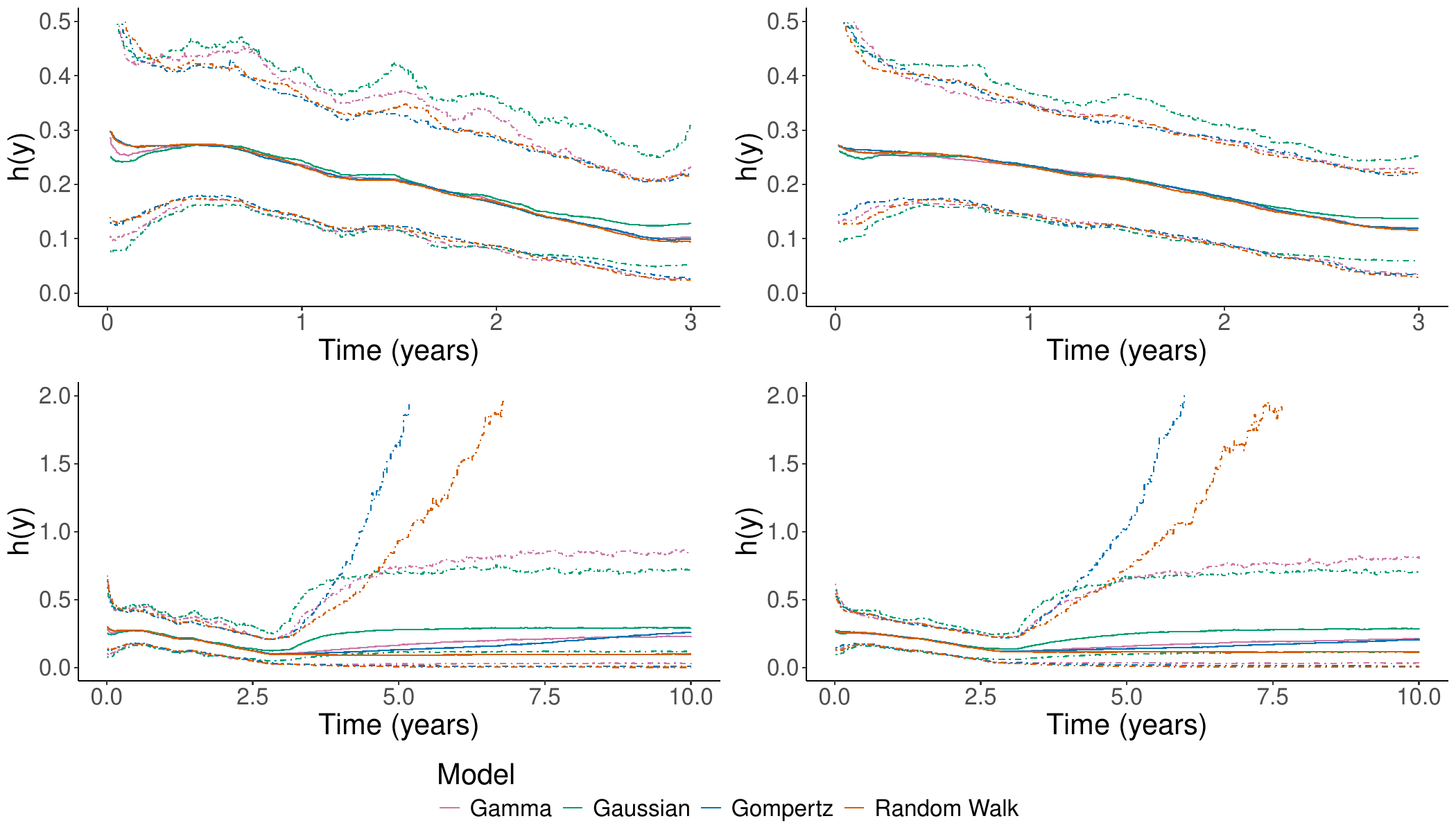}
    \caption{Hazard functions for the colon cancer data for observation period (top) and total period of interest (bottom) under the Poisson (left) and Negative Binomial (right) prior specifications. Median hazard values (solid) and corresponding 95\% credible intervals (dashed) are reported for varying specifications of $\mu(\alpha_j)$.}
    \label{fig:Colon_haz}
\end{figure}

\subsubsection{Alternative approaches}

To contextualise the inferences obtained under the diffusion piecewise exponential model, we consider three alternative methods: \textit{i)} The standard approach of selecting a two-parameter parametric model using information criteria (in this case the log-Normal parametric model) \citep{Latimer2011, Baio2020}. \textit{ii)} The piecewise exponential model with independent priors, where the hazard at the end of the observation period is taken as the hazard for the extrapolation period \citep{Cooney2023a}. \textit{iii)} Modelling the hazard using M-splines \citep{Jackson2023}. In particular, as extrapolations are based on the placement of a final knot on $(y_+, y_\infty)$, we consider inferences under three different knot locations. Full implementation details and additional analysis are provided in the supplementary materials \citep{Hardcastle2024supp}. Mean survival estimates are reported in Table \ref{tab:Colon}. 

In each case mean survival estimates for the observation period are close to those reported by the diffusion piecewise exponential model, although the spline and independent piecewise model provide slightly larger estimates of mean survival. We expect this to be due to the influence of $\mu(\alpha_j)$ at the end of the observation period when less data are available. Total mean survival estimates vary significantly between models. Note that the log-normal reports the smallest credible intervals, as the hazard in the extrapolation period inherits the parameter uncertainty from the observation period, and is therefore underestimating the uncertainty associated in total mean survival.

Both the independent piecewise model and the M-spline models report far higher values of mean survival in the extrapolation period. As the independent piecewise model extrapolates a constant hazard from the end of the observation period, this estimate is large with smaller credible intervals than the diffusion piecewise exponential model, as there is no additional uncertainty associated with the hazard as $y\to\infty$. Under the M-spline model, the uncertainty associated with the hazard grows until the final knot, after which a constant hazard is extrapolated. As evidenced in the estimates reported in Table \ref{tab:Colon} the placement of this final knot is highly influential, yet it is unclear how this knot should be placed beyond trial and error.

\subsection{Time varying drifts}
\label{sec;TimeVarying}
In the preceding Section we have only considered time homogeneous drift functions to guide extrapolations. As observed in Section \ref{sec;diff_prior}, however, the prior structure can naturally be extended to incorporate time-varying drifts, $\mu(\alpha_j,y)$. This allows for a far more expressive range of expert information to be encoded into the prior.

\subsubsection{Example time-varying drifts}

For the log-baseline hazard we consider two time-varying drifts
\begin{gather}
    \label{eq;Converge}
    \mu(\alpha_j, y) = \psi_1(y) - \psi_2(y)\exp(\alpha_j), \\ 
    \label{eq;Centred}
    \mu(\alpha_j, y) = \frac{1}{\psi_2^2}(\alpha_j - \psi_1(y)),
\end{gather}
constructed by adding time-varying parameters into the two Langevin drifts considered previously. In particular for the first drift $(\psi_1(y), \psi_2(y))$ are constructed such that they taper between the parameters of two different Gamma distributions on a finite given interval. As such this drift encodes a highly informative prior about the long-term hazard, but a weaker prior to be used for the observation period. The second drift allows the prior mean of the log-hazard to vary with time. In particular this allows for a pre-specified hazard function, for example elicited from previous clinical trials, to be used to guide long-term extrapolations. We note that this combining of the observed hazard with a pre-specified long-term hazard bears a strong resemblance to the blended survival approach of \citep{Che2023}, albeit on the hazard rather than survival function.

A similar consideration can be taken when incorporating covariates directly in the model. Often in these cases, analysts will seek to encode a waning treatment effect assumption into extrapolations \citep{Jackson2017}. This can be done explicitly within our framework as
\begin{equation}
    \label{eq;trtwane}
    \mu(\beta_j, y) = \frac{1}{\psi_2(y)^2}\beta_j,
\end{equation}
shrinking the treatment effect to 0 as $y\to\infty$.

\subsubsection{CLL-8 trial data}

We apply the time-varying drifts to data from the CLL-8 trial \citep{Williams2017}, that investigated the effect of an immunotherapy treatment in combination with chemotherapy on survival in chronic lymphocytic leukemia (CLL) patients, compared to survival in patients who received chemotherapy alone.  810 patients were enrolled with 403 randomised to the treatment group and 407 to the control group, with only 11.5\% of patients dying during the trial. Previous analysis has noted that there is expected to be a notable drop in $S(y)$ after 4 years \citep{Che2023}. In particular we compare a Langevin diffusion prior with a fixed Gamma(10,10) distribution to one that converges to a Gamma(10,10) distribution in the extrapolation period, specified by \eqref{eq;Converge}. We also compare the baseline Gompertz drift to another centred around a given Gompertz hazard function \eqref{eq;Centred}. Generating two chains of 10,000 samples including burn-in took under 2 minutes for each model, except for the Gamma(10,10) model where poor prior specification hindered computation. Full prior specifications, computational details and further results are provided in the supplementary materials \citep{Hardcastle2024supp}.

\begin{figure}[]
    \centering
    \includegraphics[width=\linewidth]{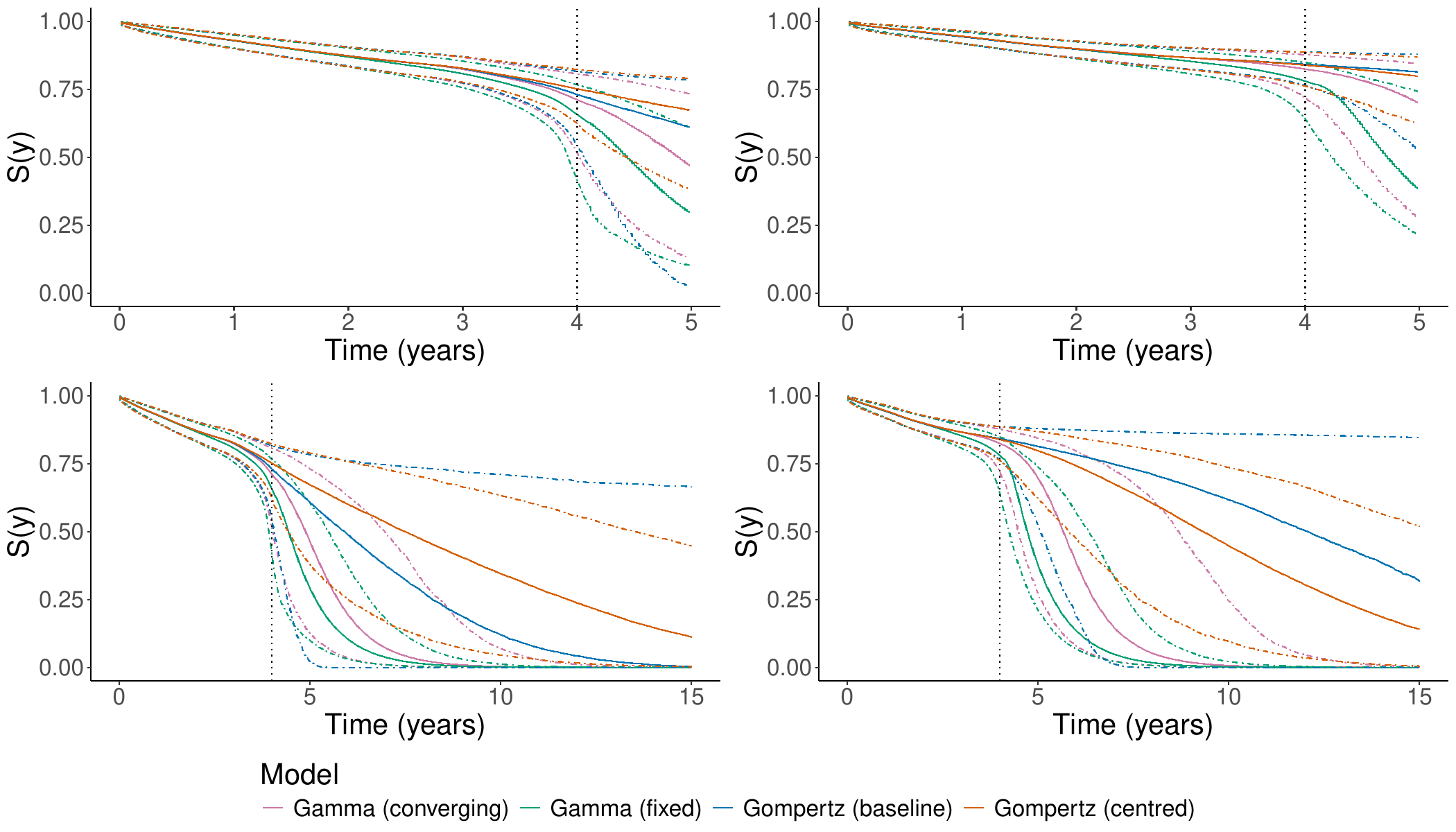}
    \caption{Survival curves for the diffusion piecewise exponential model for varying specifications of $\mu(\alpha_j)$ fit to the control (left) and treatment (right) arms of the CLL-8 trial data. Curves are plotted for the observation period (top) and extrapolation period (bottom), with $y_+ = 4$ denoted by the dotted line. Median values for $S(y)$ are given by the solid lines with 95\% credible intervals indicated by the dashed lines.}
    \label{fig:TA174}
\end{figure}

\begin{table}[]
    \centering
    \begin{tabular}{l|l|l|l}
       Model & Trial arm & $\mathbb{E}[Y]$ on $(0,y_+)$ & $\mathbb{E}[Y]$ on $(0,y_\infty)$\\
       \hline
        Gamma fixed  & Control  & 3.51 (3.36, 3.64) & 4.25 (3.70, 5.03) \\
        Gamma fixed  & Treatment  & 3.66	(3.53, 3.77) & 4.67	(4.14, 5.81) \\
        Gamma waning & Control & 3.55	(3.40, 3.68) & 4.71 (3.83, 6.25) \\
        Gamma waning & Treatment   & 3.70	(3.58, 3.80) & 5.54 (4.37, 7.95) \\
        Gompertz Baseline & Control & 3.56 (3.41, 3.69)	& 6.61 (3.72, 11.30) \\
        Gompertz Baseline & Treatment & 3.70 (3.58, 3.80)	& 9.78 (4.80, 13.11) \\
        Gompertz centred & Control & 3.58 (3.44, 3.71) & 10.75 (8.35, 12.18) \\
        Gompertz centred & Treatment & 3.71 (3.60, 3.82)	& 12.17 (10.65, 13.08) \\
        \hline
    \end{tabular}
    \caption{Estimates for mean survival for the CLL-8 trial in the control and treatment arms when modelled independently under various prior assumptions. Expected mean survival and 95\% credible intervals are reported for the observation period and the entire window of interest.}
    \label{tab:TA174}
\end{table}

Survival curves for the above drifts are provided in Figure \ref{fig:TA174} for both the treatment and control arms along with corresponding posterior estimates of mean survival in Table \ref{tab:TA174}. Expected mean survival is larger under each prior specification. Note that compared to the data in Section \ref{sec;Colon} events are rarer near the point of administrative censoring and therefore $\mu(\alpha_j,y)$ is more influential before $y_+$, as can be observed in Figure \ref{fig:TA174}. This effect is particularly profound for the fixed Gamma(10,10) drift, in contrast \eqref{eq;Converge} allows for the data to remain informative for longer before $\mu(\alpha_j,y)$ becomes influential. As can be seen in both trial arms, the Gompertz baseline prior is highly sensitive to the value of the survival function at $y_+$. As a result, minor differences in the data (as seen between the two trial arms) give rise to very different long-term survival estimates. In contrast, the centred hazard provides a far more controlled method for incorporating informative long-term information.

\begin{figure}
    \centering
    \includegraphics[width=\linewidth]{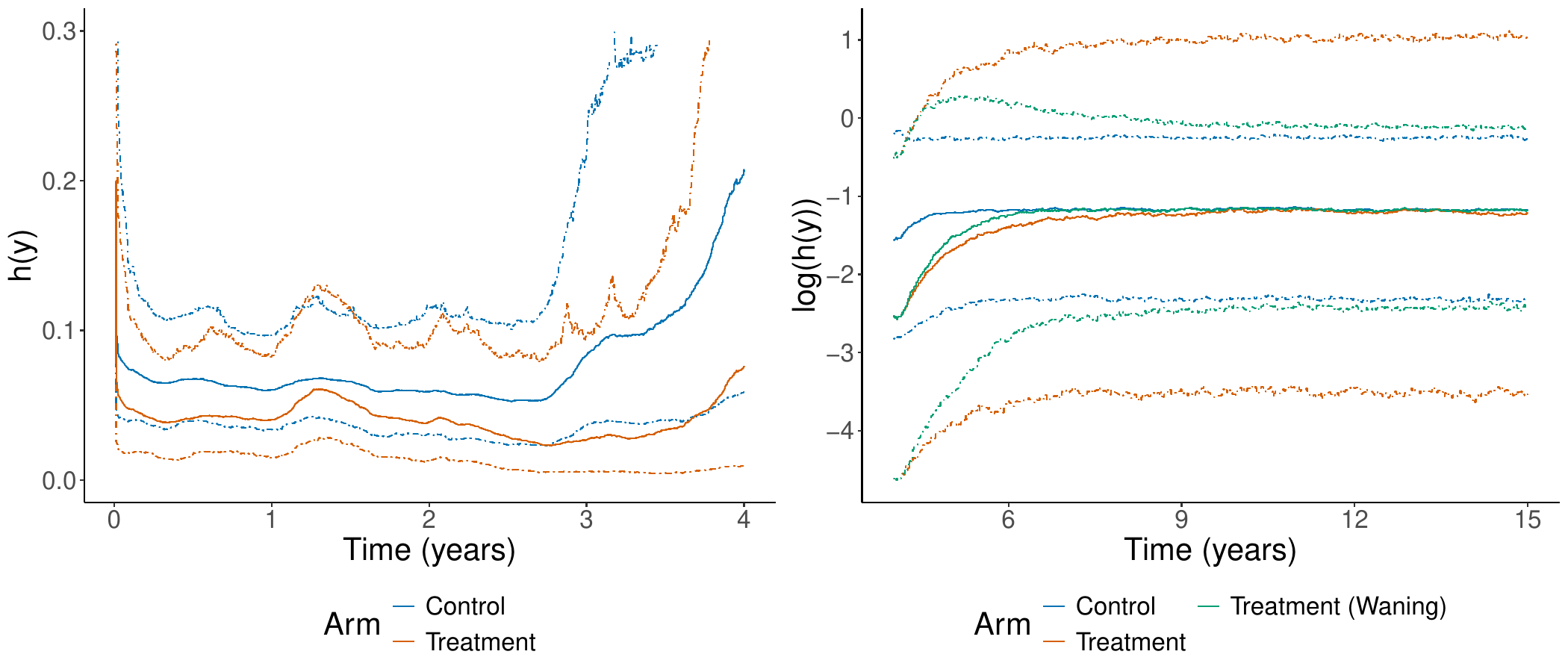}
    \caption{(Left) Hazard functions for the control and treatment arms with corresponding 95\% credible intervals during the observation period. (Right) Log-hazard functions for the control and treatment arms (under both waning and non-waning assumptions) during the extrapolation period.}
    \label{fig:TA174cov}
\end{figure}

\begin{table}[]
    \centering
    \begin{tabular}{l|l|l|l}
       Treat. arm & $\mathbb{E}[Y]$ on $(0,y_+)$ & $\mathbb{E}[Y]$ on $(0,y_\infty)$ & $\mathbb{E}[Y_{\text{t}}] - \mathbb{E}[Y_{\text{c}}]$ \\
       \hline
        Control  & 3.55 (3.41, 3.67) & 5.62 (4.84, 6.60) & --- \\
        Treatment (fixed)  & 3.73 (3.61, 3.82) & 6.34 (4.56, 9.02) & 0.73 (-1.07, 3.17) \\
        Treatment (waning) & 3.73 (3.61, 3.82) & 6.29 (4.90, 7.73) & 0.68 (-0.64, 1.87) \\
        \hline
    \end{tabular}
    \caption{Estimates for mean survival and corresponding 95\% credible intervals for the CLL-8 trial in the control and treatment arms (under both waning and non-waning assumptions) during the observation period, $(0,y_+)$, and the entire window of interest, $(0,y_\infty)$. The final column reports estimates of the difference between mean survival for the treatment and control groups.}
    \label{tab:TA174_cov}
\end{table}

We conclude by investigating the model when covariates are directly incorporated rather than modelled independently. Here we use a $\text{Gamma}(5,15)$ Langevin diffusion for the baseline log-hazard and compare a $\text{Normal}(0,1)$ Langevin drift for $\mu(\beta_{y})$, with \eqref{eq;trtwane} where the waning begins after $y_+$ resulting in identical inferences during the observation period. Hazards for the observation period and log-hazards for the extrapolation period are shown in Figure \ref{fig:TA174cov}, with mean survival estimates provided in Table \ref{tab:TA174_cov}. 

From the hazard functions, there is clear evidence of some non-proportionality in the observation period, and some weak evidence to suggest the treatment is beneficial compared to the control, that is corroborated by mean survival estimates. Examining the extrapolation period, both drifts for the treatment arm imply that the expected hazard should converge to the hazard for the control arm. For the fixed Langevin diffusion, uncertainty then arises from both the process for $\alpha_j$ and $\beta_j$. In contrast, the treatment hazard converges faster to the hazard of the control arm, and the associated credible intervals are far smaller. This is reflected in the estimates of the difference in mean survival where the waning assumption reduces the uncertainty in the estimates of difference in mean survival. We note that treatment effect waning is a strong and untestable assumption that in practice will require expert justification to be incorporated. 

\section{Discussion}
\label{sec;Discussion}
In this work we have introduced the diffusion piecewise exponential model, a novel prior structure combining flexible modelling of the hazard function in the observation period with expert information in the extrapolation period within a principled Bayesian framework.

No model can automatically guarantee plausible extrapolations. The diffusion piecewise exponential model is no exception, with reasonable extrapolation relying on sensible specification of $\mu$. Our approach has key advantages, however, compared to current state-of-the-art methods. First, as demonstrated through the variety of drifts used in Section \ref{sec;Examples}, $\mu$ is able to incorporate a wide-range of prior information, with minimal restrictions on the form this should take. Second, specification of this prior information is only weakly informative during the observation period, becoming increasingly influential as the data become sparse. Finally, the assumptions encoded into this prior are explicit and easy to interrogate. This is a core part of the process of appraising the cost-effectiveness of novel medical interventions, and as such our model promotes improved decision making and analysis by both pharmaceutical companies and regulatory bodies. 

We have presented a wide range of possibilities for the specification of $\mu$, but the examples considered here are by no means exhaustive. In the context of clinical trial data with two treatment arms, for example, dependence between each hazard could be introduced through $\mu$ rather than the local proportional hazard assumption incorporated in this work. We believe the design of drift functions that can capture an even wider range of expert information to be an exciting avenue for future research.

We have assumed a Poisson process prior for $\{s_j\}_{j=1}^J$ with homogeneous intensity. This assumption could be altered to incorporate a process with, for example, decreasing intensity if more volatility is expected at the start of the observation period. A particular strength of the sampling methods developed for this work is that changes to $\mu$ and $\gamma$ in general do not require changes to the sampler. The only weak condition for the prior on $\{s_j\}_{j=1}^J$ is the existence of a dominating process that is simple to sample from, given the thinning procedure outlined in Section \ref{sec;PDMP_knot}. 

We have shown how efficient computational procedures for sampling from posteriors induced by spike and slab priors using PDMPs can be extended to more general transdimensional posteriors. The key feature of this construction was the identification of a hyperplane in $\theta$-space such that the likelihoods of the simpler and more complex models were identical. In the reversible jump literature this is referred to as a centring point \citep{Brooks2003} and is a common feature of many transdimensional posteriors. The sampling framework provided in Section \ref{sec;Sampling} should therefore allow for the extension of sticky PDMP dynamics to a far wider range of transdimensional sampling problems.

\begin{acks}[Acknowledgments]
    LH would like to thank Sebastiano Grazzi and Sam Power for useful discussions.
\end{acks}

\begin{funding}
LH was supported by EPSRC grant EP/W523835/1.
\end{funding}

\begin{supplement}
\stitle{Supplementary materials for Diffusion piecewise exponential models for survival extrapolation using Piecewise Deterministic Monte Carlo}
\sdescription{Additional derivations for Section \ref{sec;Model}; Algorithms for Section \ref{sec;Sampling} and additional computational details. Further modelling details for Section \ref{sec;Examples} and full details of comparator models.}
\end{supplement}

\bibliographystyle{ba}
\bibliography{TheBib}

\begin{thebibliography}{54}
\newcommand{\enquote}[1]{``#1''}
\expandafter\ifx\csname natexlab\endcsname\relax\def\natexlab#1{#1}\fi
\expandafter\ifx\csname url\endcsname\relax
  \def\url#1{{\tt #1}}\fi
\expandafter\ifx\csname urlprefix\endcsname\relax\def\urlprefix{URL }\fi
\ifx\endbibitem\undefined \let\endbibitem\relax\fi

\bibitem[{Aalen and Gjessing(2004)}]{Aalen2004}
Aalen, O.~O. and Gjessing, H.~K. (2004).
\newblock \enquote{Survival models based on the Ornstein-Uhlenbeck process.}
\newblock {\em Lifetime data analysis\/}, 10: 407--423.
\endbibitem

\bibitem[{Andral and Kamatani(2024)}]{Andral2024}
Andral, C. and Kamatani, K. (2024).
\newblock \enquote{Automated Techniques for Efficient Sampling of Piecewise-Deterministic Markov Processes.}
\newblock {\em arXiv preprint arXiv:2408.03682\/}.
\endbibitem

\bibitem[{Andrieu and Livingstone(2021)}]{Andrieu2021}
Andrieu, C. and Livingstone, S. (2021).
\newblock \enquote{Peskun--Tierney ordering for Markovian Monte Carlo: beyond the reversible scenario.}
\newblock {\em The Annals of Statistics\/}, 49(4): 1958--1981.
\endbibitem

\bibitem[{Bagust and Beale(2014)}]{Bagust2014}
Bagust, A. and Beale, S. (2014).
\newblock \enquote{Survival Analysis and Extrapolation Modeling of Time-to-Event Clinical Trial Data for Economic Evaluation: An Alternative Approach.}
\newblock {\em Medical Decision Making\/}, 34(3): 343--351.
\newblock PMID: 23901052.
\endbibitem

\bibitem[{Baio(2020)}]{Baio2020}
Baio, G. (2020).
\newblock \enquote{survHE: survival analysis for health economic evaluation and cost-effectiveness modeling.}
\newblock {\em Journal of Statistical Software\/}, 95: 1--47.
\endbibitem

\bibitem[{Bertazzi and Bierkens(2022)}]{Bertazzi2022}
Bertazzi, A. and Bierkens, J. (2022).
\newblock \enquote{Adaptive schemes for piecewise deterministic Monte Carlo algorithms.}
\newblock {\em Bernoulli\/}, 28(4): 2404--2430.
\endbibitem

\bibitem[{Bertazzi et~al.(2023)Bertazzi, Dobson, and Monmarch{\'e}}]{Bertazzi2023}
Bertazzi, A., Dobson, P., and Monmarch{\'e}, P. (2023).
\newblock \enquote{Piecewise deterministic sampling with splitting schemes.}
\newblock {\em arXiv preprint arXiv:2301.02537\/}.
\endbibitem

\bibitem[{Betancourt and Girolami(2015)}]{Betancourt2015}
Betancourt, M. and Girolami, M. (2015).
\newblock \enquote{Hamiltonian Monte Carlo for hierarchical models.}
\newblock {\em Current trends in Bayesian methodology with applications\/}, 79(30): 2--4.
\endbibitem

\bibitem[{Bierkens et~al.(2023)Bierkens, Grazzi, Meulen, and Schauer}]{Bierkens2023a}
Bierkens, J., Grazzi, S., Meulen, F. v.~d., and Schauer, M. (2023).
\newblock \enquote{Sticky PDMP samplers for sparse and local inference problems.}
\newblock {\em Statistics and Computing\/}, 33(1): 8.
\endbibitem

\bibitem[{Bouchard-C{\^o}t{\'e} et~al.(2018)Bouchard-C{\^o}t{\'e}, Vollmer, and Doucet}]{Bouchard2018}
Bouchard-C{\^o}t{\'e}, A., Vollmer, S.~J., and Doucet, A. (2018).
\newblock \enquote{The bouncy particle sampler: A nonreversible rejection-free Markov chain Monte Carlo method.}
\newblock {\em Journal of the American Statistical Association\/}, 113(522): 855--867.
\endbibitem

\bibitem[{Brooks et~al.(2003)Brooks, Giudici, and Roberts}]{Brooks2003}
Brooks, S.~P., Giudici, P., and Roberts, G.~O. (2003).
\newblock \enquote{Efficient construction of reversible jump Markov chain Monte Carlo proposal distributions.}
\newblock {\em Journal of the Royal Statistical Society Series B: Statistical Methodology\/}, 65(1): 3--39.
\endbibitem

\bibitem[{Chapple et~al.(2020)Chapple, Peak, and Hemal}]{Chapple2020}
Chapple, A.~G., Peak, T., and Hemal, A. (2020).
\newblock \enquote{A novel Bayesian continuous piecewise linear log-hazard model, with estimation and inference via reversible jump Markov chain Monte Carlo.}
\newblock {\em Statistics in medicine\/}, 39(12): 1766--1780.
\endbibitem

\bibitem[{Che et~al.(2023)Che, Green, and Baio}]{Che2023}
Che, Z., Green, N., and Baio, G. (2023).
\newblock \enquote{Blended survival curves: a new approach to extrapolation for time-to-event outcomes from clinical trials in health technology assessment.}
\newblock {\em Medical Decision Making\/}, 43(3): 299--310.
\endbibitem

\bibitem[{Chevallier et~al.(2023)Chevallier, Fearnhead, and Sutton}]{Chevallier2023}
Chevallier, A., Fearnhead, P., and Sutton, M. (2023).
\newblock \enquote{Reversible jump PDMP samplers for variable selection.}
\newblock {\em Journal of the American Statistical Association\/}, 118(544): 2915--2927.
\endbibitem

\bibitem[{Cooney and White(2023)}]{Cooney2023a}
Cooney, P. and White, A. (2023).
\newblock \enquote{Extending Beyond Bagust and Beale: Fully Parametric Piecewise Exponential Models for Extrapolation of Survival Outcomes in Health Technology Assessment.}
\newblock {\em Value in Health\/}, 26(10): 1510--1517.
\endbibitem

\bibitem[{Corbella et~al.(2022)Corbella, Spencer, and Roberts}]{Corbella2022}
Corbella, A., Spencer, S.~E., and Roberts, G.~O. (2022).
\newblock \enquote{Automatic Zig-Zag sampling in practice.}
\newblock {\em Statistics and Computing\/}, 32(6): 107.
\endbibitem

\bibitem[{Demarqui et~al.(2012)Demarqui, Loschi, Dey, and Colosimo}]{Demarqui2012}
Demarqui, F.~N., Loschi, R.~H., Dey, D.~K., and Colosimo, E.~A. (2012).
\newblock \enquote{A class of dynamic piecewise exponential models with random time grid.}
\newblock {\em Journal of Statistical Planning and Inference\/}, 142(3): 728--742.
\endbibitem

\bibitem[{Demiris et~al.(2015)Demiris, Lunn, and Sharples}]{Demiris2015}
Demiris, N., Lunn, D., and Sharples, L.~D. (2015).
\newblock \enquote{Survival extrapolation using the poly-Weibull model.}
\newblock {\em Statistical Methods in Medical Research\/}, 24(2): 287--301.
\endbibitem

\bibitem[{Fahrmeir and Lang(2001)}]{Fahrmeir2001}
Fahrmeir, L. and Lang, S. (2001).
\newblock \enquote{Bayesian inference for generalized additive mixed models based on Markov random field priors.}
\newblock {\em Journal of the Royal Statistical Society Series C: Applied Statistics\/}, 50(2): 201--220.
\endbibitem

\bibitem[{Fearnhead et~al.(2024)Fearnhead, Nemeth, Oates, and Sherlock}]{Fearnhead2024b}
Fearnhead, P., Nemeth, C., Oates, C.~J., and Sherlock, C. (2024).
\newblock \enquote{Scalable Monte Carlo for Bayesian Learning.}
\newblock {\em arXiv preprint arXiv:2407.12751\/}.
\endbibitem

\bibitem[{Feigl and Zelen(1965)}]{Feigl1965}
Feigl, P. and Zelen, M. (1965).
\newblock \enquote{Estimation of exponential survival probabilities with concomitant information.}
\newblock {\em Biometrics\/}, 826--838.
\endbibitem

\bibitem[{Gibbons and Latimer(2024)}]{Gibbons2024}
Gibbons, C.~L. and Latimer, N.~R. (2024).
\newblock \enquote{Prevalence of Immature Survival Data for Anti-Cancer Drugs Presented to the National Institute for Health and Care Excellence between 2018-2022.}
\newblock {\em Value in Health\/}.
\endbibitem

\bibitem[{Green(1995)}]{Green1995}
Green, P.~J. (1995).
\newblock \enquote{Reversible jump Markov chain Monte Carlo computation and Bayesian model determination.}
\newblock {\em Biometrika\/}, 82(4): 711--732.
\endbibitem

\bibitem[{Hardcastle et~al.(2024{\natexlab{a}})Hardcastle, Livingstone, and Baio}]{Hardcastle2024}
Hardcastle, L., Livingstone, S., and Baio, G. (2024{\natexlab{a}}).
\newblock \enquote{Averaging polyhazard models using Piecewise deterministic Monte Carlo with applications to data with long-term survivors.}
\newblock {\em arXiv preprint arXiv:2406.14182\/}.
\endbibitem

\bibitem[{Hardcastle et~al.(2024{\natexlab{b}})Hardcastle, Livingstone, and Baio}]{Hardcastle2024supp}
--- (2024{\natexlab{b}}).
\newblock \enquote{Supplement to ``Diffusion piecewise exponential models for survival extrapolation using Piecewise Deterministic Monte Carlo".}
\endbibitem

\bibitem[{Hird et~al.(2020)Hird, Livingstone, and Zanella}]{Hird2020}
Hird, M., Livingstone, S., and Zanella, G. (2020).
\newblock \enquote{A fresh take on ‘Barker dynamics’ for MCMC.}
\newblock In {\em International Conference on Monte Carlo and Quasi-Monte Carlo Methods in Scientific Computing\/}, 169--184. Springer.
\endbibitem

\bibitem[{Ibrahim et~al.(2001)Ibrahim, Chen, Sinha, Ibrahim, and Chen}]{Ibrahim2001}
Ibrahim, J.~G., Chen, M.-H., Sinha, D., Ibrahim, J., and Chen, M. (2001).
\newblock {\em Bayesian survival analysis\/}, volume~2.
\newblock Springer.
\endbibitem

\bibitem[{Jackson et~al.(2017)Jackson, Stevens, Ren, Latimer, Bojke, Manca, and Sharples}]{Jackson2017}
Jackson, C., Stevens, J., Ren, S., Latimer, N., Bojke, L., Manca, A., and Sharples, L. (2017).
\newblock \enquote{Extrapolating Survival from Randomized Trials Using External Data: A Review of Methods.}
\newblock {\em Medical Decision Making\/}, 37(4): 377--390.
\newblock PMID: 27005519.
\endbibitem

\bibitem[{Jackson(2023)}]{Jackson2023}
Jackson, C.~H. (2023).
\newblock \enquote{survextrap: a package for flexible and transparent survival extrapolation.}
\newblock {\em BMC Medical Research Methodology\/}, 23(1): 282.
\endbibitem

\bibitem[{Kearns et~al.(2019)Kearns, Stevenson, Triantafyllopoulos, and Manca}]{Kearns2019}
Kearns, B., Stevenson, M.~D., Triantafyllopoulos, K., and Manca, A. (2019).
\newblock \enquote{Generalized linear models for flexible parametric modeling of the hazard function.}
\newblock {\em Medical Decision Making\/}, 39(7): 867--878.
\endbibitem

\bibitem[{Kearns et~al.(2022)Kearns, Stevenson, Triantafyllopoulos, and Manca}]{Kearns2022}
--- (2022).
\newblock \enquote{Dynamic and Flexible Survival Models for Extrapolation of Relative Survival: A Case Study and Simulation Study.}
\newblock {\em Medical Decision Making\/}, 42(7): 945--955.
\endbibitem

\bibitem[{Latimer(2011)}]{Latimer2011}
Latimer, N. (2011).
\newblock \enquote{NICE DSU technical support document 14: survival analysis for economic evaluations alongside clinical trials-extrapolation with patient-level data.}
\newblock {\em Report by the Decision Support Unit\/}.
\endbibitem

\bibitem[{Latimer(2014)}]{Latimer2014}
Latimer, N.~R. (2014).
\newblock \enquote{Response to “survival analysis and extrapolation modeling of time-to-event clinical trial data for economic evaluation: an alternative approach” by Bagust and Beale.}
\newblock {\em Medical Decision Making\/}, 34(3): 279--282.
\endbibitem

\bibitem[{Lin et~al.(2021)Lin, Thall, and Yuan}]{Lin2021}
Lin, R., Thall, P.~F., and Yuan, Y. (2021).
\newblock \enquote{Bags: A Bayesian adaptive group sequential trial design with subgroup-specific survival comparisons.}
\newblock {\em Journal of the American Statistical Association\/}, 116(533): 322--334.
\endbibitem

\bibitem[{Livingstone et~al.(2019)Livingstone, Faulkner, and Roberts}]{Livingstone2019}
Livingstone, S., Faulkner, M.~F., and Roberts, G.~O. (2019).
\newblock \enquote{Kinetic energy choice in Hamiltonian/hybrid Monte Carlo.}
\newblock {\em Biometrika\/}, 106(2): 303--319.
\endbibitem

\bibitem[{Livingstone et~al.(2024)Livingstone, N{\"u}sken, Vasdekis, and Zhang}]{Livingstone2024}
Livingstone, S., N{\"u}sken, N., Vasdekis, G., and Zhang, R.-Y. (2024).
\newblock \enquote{Skew-symmetric schemes for stochastic differential equations with non-Lipschitz drift: an unadjusted Barker algorithm.}
\newblock {\em arXiv preprint arXiv:2405.14373\/}.
\endbibitem

\bibitem[{Livingstone and Zanella(2022)}]{Livingstone2022}
Livingstone, S. and Zanella, G. (2022).
\newblock \enquote{The Barker proposal: combining robustness and efficiency in gradient-based MCMC.}
\newblock {\em Journal of the Royal Statistical Society Series B: Statistical Methodology\/}, 84(2): 496--523.
\endbibitem

\bibitem[{Michel et~al.(2020)Michel, Durmus, and S{\'e}n{\'e}cal}]{Michel2020}
Michel, M., Durmus, A., and S{\'e}n{\'e}cal, S. (2020).
\newblock \enquote{Forward event-chain Monte Carlo: Fast sampling by randomness control in irreversible Markov chains.}
\newblock {\em Journal of Computational and Graphical Statistics\/}, 29(4): 689--702.
\endbibitem

\bibitem[{Mikkola et~al.(2023)Mikkola, Martin, Chandramouli, Hartmann, Pla, Thomas, Pesonen, Corander, Vehtari, Kaski et~al.}]{Mikkola2024}
Mikkola, P., Martin, O.~A., Chandramouli, S., Hartmann, M., Pla, O.~A., Thomas, O., Pesonen, H., Corander, J., Vehtari, A., Kaski, S., et~al. (2023).
\newblock \enquote{Prior Knowledge Elicitation: The Past, Present, and Future.}
\newblock {\em Bayesian Analysis\/}, 1(1): 1--33.
\endbibitem

\bibitem[{Murray et~al.(2016)Murray, Hobbs, Sargent, and Carlin}]{Murray2016}
Murray, T.~A., Hobbs, B.~P., Sargent, D.~J., and Carlin, B.~P. (2016).
\newblock \enquote{Flexible Bayesian survival modeling with semiparametric time-dependent and shape-restricted covariate effects.}
\newblock {\em Bayesian analysis (Online)\/}, 11(2): 381.
\endbibitem

\bibitem[{Oakley et~al.(2025)Oakley, Ren, Forsyth, Gosling, Wilson, Latimer, Rutherford, Uttley, and Fotheringham}]{Oakley2025}
Oakley, J.~E., Ren, S., Forsyth, J.~E., Gosling, J.~P., Wilson, K., Latimer, N., Rutherford, M.~J., Uttley, L., and Fotheringham, J. (2025).
\newblock \enquote{NICE DSU Technical Support Document 26: Expert elicitation for long-term survival outcomes.}
\newblock Technical Support Document~26, Decision Support Unit, National Institute for Health and Care Excellence (NICE).
\endbibitem

\bibitem[{Oksendal(2013)}]{Oksendal2013}
Oksendal, B. (2013).
\newblock {\em Stochastic differential equations: an introduction with applications\/}.
\newblock Springer Science \& Business Media.
\endbibitem

\bibitem[{Palmer et~al.(2023)Palmer, Lin, Martin, Jagannath, Jakubowiak, Usmani, Buyukkaramikli, Phelps, Slowik, Pan et~al.}]{Palmer2023}
Palmer, S., Lin, Y., Martin, T.~G., Jagannath, S., Jakubowiak, A., Usmani, S.~Z., Buyukkaramikli, N., Phelps, H., Slowik, R., Pan, F., et~al. (2023).
\newblock \enquote{Extrapolation of survival data using a Bayesian approach: a case study leveraging external data from Cilta-Cel therapy in multiple myeloma.}
\newblock {\em Oncology and Therapy\/}, 11(3): 313--326.
\endbibitem

\bibitem[{Platen and Bruti-Liberati(2010)}]{Platen2010}
Platen, E. and Bruti-Liberati, N. (2010).
\newblock {\em Numerical solution of stochastic differential equations with jumps in finance\/}, volume~64.
\newblock Springer Science \& Business Media.
\endbibitem

\bibitem[{Roberts and Sangalli(2010)}]{Roberts2010}
Roberts, G.~O. and Sangalli, L.~M. (2010).
\newblock \enquote{{Latent diffusion models for survival analysis}.}
\newblock {\em Bernoulli\/}, 16(2): 435 -- 458.
\endbibitem

\bibitem[{Roberts and Tweedie(1996)}]{Roberts1996}
Roberts, G.~O. and Tweedie, R.~L. (1996).
\newblock \enquote{Exponential convergence of Langevin distributions and their discrete approximations.}
\newblock {\em Bernoulli\/}, 2(3): 341--363.
\endbibitem

\bibitem[{Rutherford et~al.(2020)Rutherford, Lambert, Sweeting, Pennington, Crowther, Abrams et~al.}]{Rutherford2020}
Rutherford, M.~J., Lambert, P.~C., Sweeting, M.~J., Pennington, R., Crowther, M.~J., Abrams, K.~R., et~al. (2020).
\newblock \enquote{NICE DSU technical support document 21: flexible methods for survival analysis.}
\newblock {\em Decision Support Unit, ScHARR, University of Sheffield\/}.
\endbibitem

\bibitem[{Sachs et~al.(2023)Sachs, Sen, Lu, and Dunson}]{Sachs2023}
Sachs, M., Sen, D., Lu, J., and Dunson, D. (2023).
\newblock \enquote{Posterior computation with the Gibbs zig-zag sampler.}
\newblock {\em Bayesian Analysis\/}, 18(3): 909--927.
\endbibitem

\bibitem[{Sharef et~al.(2010)Sharef, Strawderman, Ruppert, Cowen, and Halasyamani}]{Sharef2010}
Sharef, E., Strawderman, R.~L., Ruppert, D., Cowen, M., and Halasyamani, L. (2010).
\newblock \enquote{Bayesian adaptive B-spline estimation in proportional hazards frailty models.}
\newblock {\em Electronic Journal of Statistics\/}, 4: 606--642.
\endbibitem

\bibitem[{Simpson et~al.(2017)Simpson, Rue, Riebler, Martins, and S{\o}rbye}]{Simpson2017}
Simpson, D., Rue, H., Riebler, A., Martins, T.~G., and S{\o}rbye, S.~H. (2017).
\newblock \enquote{{Penalising Model Component Complexity: A Principled, Practical Approach to Constructing Priors}.}
\newblock {\em Statistical Science\/}, 32(1): 1 -- 28.
\endbibitem

\bibitem[{Sutton and Fearnhead(2023)}]{Sutton2023}
Sutton, M. and Fearnhead, P. (2023).
\newblock \enquote{Concave-convex PDMP-based sampling.}
\newblock {\em Journal of Computational and Graphical Statistics\/}, 32(4): 1425--1435.
\endbibitem

\bibitem[{Thatcher(1999)}]{Thatcher1999}
Thatcher, A.~R. (1999).
\newblock \enquote{The long-term pattern of adult mortality and the highest attained age.}
\newblock {\em Journal of the Royal Statistical Society: Series A (Statistics in Society)\/}, 162(1): 5--43.
\endbibitem

\bibitem[{Vehtari et~al.(2017)Vehtari, Gelman, and Gabry}]{Vehtari2017}
Vehtari, A., Gelman, A., and Gabry, J. (2017).
\newblock \enquote{Practical Bayesian model evaluation using leave-one-out cross-validation and WAIC.}
\newblock {\em Statistics and computing\/}, 27: 1413--1432.
\endbibitem

\bibitem[{Williams et~al.(2017)Williams, Lewsey, Mackay, and Briggs}]{Williams2017}
Williams, C., Lewsey, J.~D., Mackay, D.~F., and Briggs, A.~H. (2017).
\newblock \enquote{Estimation of survival probabilities for use in cost-effectiveness analyses: a comparison of a multi-state modeling survival analysis approach with partitioned survival and Markov decision-analytic modeling.}
\newblock {\em Medical Decision Making\/}, 37(4): 427--439.
\endbibitem

\end{thebibliography}

\end{document}